\documentclass[12pt]{article}
\usepackage{amsthm,amsmath,amsfonts,amssymb,cases}
\newcommand{\R}{{\mathord{\mathbb R}}}
\newcommand{\Z}{{\mathord{\mathbb Z}}}
\newcommand{\N}{{\mathord{\mathbb N}}}
\newcommand{\C}{{\mathord{\mathbb C}}}

\def\chib {\overline{\chi}}

\newcommand{\HH}{\mathcal{H}}

\newcommand{\FF}{\mathcal{F}}

\newcommand{\WW}{\mathcal{W}}
\newcommand{\hh}{\mathfrak{h}}
\newcommand{\UU}{\mathcal{U}}

\newcommand{\umm}{\underline{m}}
\newcommand{\unn}{\underline{n}}
\newcommand{\upp}{\underline{p}}
\newcommand{\uqq}{\underline{q}}
\newcommand{\uzz}{\underline{0}}

\newcommand{\ran}{{\rm Ran}}

\newcommand{\un}[1]{\underline{#1}}

\DeclareMathOperator*{\esssup}{ess\,sup}

\newcommand{\ben}{\begin{displaymath}}
\newcommand{\een}{\end{displaymath}}
\newcommand{\beqn}{\begin{equation}}
\newcommand{\eeqn}{\end{equation}}
\newcommand{\beqna}{\begin{eqnarray*}}
\newcommand{\eeqna}{\end{eqnarray*}}

\def\inf{{\rm inf}\,}

\def\supp{\operatorname{supp}}

\newcommand{\sfrac}[2]{\textrm{\footnotesize $\frac{#1}{#2}$}}

\newtheorem{lemma}{Lemma}
\newtheorem{theorem}[lemma]{Theorem}
\newtheorem{remark}[lemma]{Remark}

\newtheorem{corollary}[lemma]{Corollary}
\newtheorem{definition}[lemma]{Definition}

\begin{document}
\title{Convergent expansions in non-relativistic QED: Analyticity of the ground state}
\author{\vspace{5pt} D. Hasler $^1$\footnote{
E-mail: david.hasler@math.lmu.de, on leave from College of William \& Mary} and I.
Herbst$^2$\footnote{E-mail: iwh@virginia.edu.} \\
\vspace{-4pt} \small{$1.$ Department of Mathematics,
Ludwig Maximilians University,} \\ \small{Munich, Germany}\\
\vspace{-4pt}
\small{$2.$ Department of Mathematics, University of Virginia,} \\
\small{Charlottesville, VA, 
 USA}\\}
\date{}
\maketitle

\begin{abstract}
We consider the ground state of an atom in the framework of non-relativistic qed.
We show that the ground state as
well as the ground state energy are analytic functions of the coupling constant
which couples to the vector potential, under the assumption that the atomic Hamiltonian has a non-degenerate ground state.
Moreover, we show that the corresponding
expansion coefficients are precisely the coefficients of the associated
Raleigh-Schr\"odinger series.
As a corollary we obtain that in a scaling limit
where the ultraviolet cutoff is of the order of the Rydberg energy
 the ground state and the ground state energy have
convergent power series expansions in the fine structure constant $\alpha$, with
$\alpha$ dependent coefficients  which are finite for $\alpha \geq 0$.
\end{abstract}

\section{Introduction}

Non-relativistic quantum electrodynamics (qed) is a mathematically rigorous theory describing
low energy phenomena of matter interacting with quantized radiation.
This theory allows a mathematically rigorous treatment of
various physical aspects, see for example \cite{S04} and references therein.

In this paper we investigate expansions of the ground state and the ground state energy of an atom
as functions of the coupling constant, $g$, which couples to the vector potential of the
quantized electromagnetic field.
Such an  expansion  carries the
physical structure originating from the interactions of  bound electrons with photons.
These interactions lead to  radiative corrections and
were shown \cite{bet47} to contribute to the Lamb shift \cite{lamret47}.
The main result of this paper, Theorem   \ref{thm:main1}, shows that
the ground state as well as the ground state energy of the atom are analytic 
functions of the coupling constant $g$. We do not impose any  infrared regularization (as was needed in \cite{GH09}).
We assume that the electrons of the atom are spin-less and that the atomic
Hamiltonian has a unique ground state.
Moreover, we show that the corresponding expansion
coefficients can be calculated using
Raleigh-Schr\"odinger  perturbation theory.
To see  this we introduce an infrared cutoff $\sigma \geq 0$ and
show that the ground state as well as the ground state energy are continuous
as a function of  $\sigma$.
This permits the calculation of radiative corrections to the ground
state as well as the ground state energy to any order in the coupling constant.
To obtain  contributions of processes involving $n$ photons, one needs to expand
at least to the order $n$ in the coupling constant $g$.
The main theorem of this paper can be used to justify  a rigorous
investigation
of  ground states as well as ground state energies by means of analytic perturbation theory.

As a corollary of the main result  we  obtain
a convergent  expansion  in the fine structure constant
$\alpha$, as $\alpha$  tends to zero,
in a scaling limit where the ultraviolet cutoff is of the order of the Rydberg
energy. To this end we introduce a parameter, $\beta$,
which originates from the coupling to the electrostatic potential,
show that all estimates are uniform in $\beta$, and set  $g = \alpha^{3/2}$ and $\beta = \alpha$.
As a result, Corollary  \ref{cor:main2},
 we obtain that  the ground state and the ground state energy have
convergent power series expansions in the fine structure constant $\alpha$, with
$\alpha$ dependent coefficients which are finite for $\alpha \geq 0$.
These coefficients can be calculated by means of Raleigh-Schr\"odinger
perturbation theory.
The expansion of the ground state is in powers of $\alpha^{3/2}$ and the expansion
of the ground state energy is in powers of $\alpha^3$.
This result improves the main theorem stated in   \cite{BFP06,BFP09} where it was shown
 that there exists an asymptotic expansion in $\alpha$ involving
coefficients which depend on $\alpha$ and have at most
mild singularities.  We want to note that in
 different scaling limits
of the ultraviolet cutoff expansions in the first few orders in $\alpha$ were
obtained  in  \cite{BCVV09,HHS05,HS02}, which  involve  logarithmic terms.
The scaling limit which we consider in this paper (where the ultraviolet cutoff is of the order
of the Rydberg energy)  is typically used to study the properties of atoms, c.f. \cite{BFP06,BFP07,BFP09,HHH08,BFS99}.
In \cite{HHH08,BFS99} estimates on lifetimes of metastable states  were proven, which, in leading order, agree with experiment. 

Let us now address the proof of the main results.
It is well known that the ground state energy is embedded in the continuous
spectrum. In such a situation regular perturbation theory is
typically not applicable and other methods have to be employed.
To prove the existence result as well as the analyticity result
 we use a variant of the operator theoretic renormalization analysis as introduced
in \cite{BFS98} and further developed in \cite{BCFS03}.
The main idea of the proof is that by rotation invariance one can
infer that in the renormalization analysis terms which are linear in
creation and annihilation operators do not occur. In that
case it follows that the renormalization transformation is
a contraction even without infrared regularization.
A similar
idea was used to prove the existence and the analyticity of the ground state and the ground state energy
in the spin-boson model \cite{HH10}. In the proof we will use
results which were obtained in \cite{HH10}.
We note that related ideas were also used in \cite{GH09}.
Furthermore, we think that the method of combining the renormalization transformation
with rotation invariance, as used in this paper, might be applicable to other spectral problems of
atoms in the framework of non-relativistic qed. We note that contraction of the renormalization transformation
can also be shown using a  generalized Pauli-Fierz transformation  \cite{S07}.
As opposed to the latter reference and all other treatments we are aware of, we do not
use (or need) gauge invariance of the Hamiltonian.  Thus for example the quadratic term in the
vector potential could be dropped and our results would remain the same.

\section{Model and Statement of Results}

Let $(\hh, \langle \cdot , \cdot \rangle_\hh)$ be a
Hilbert space. We introduce the direct sum of the $n$-fold tensor product of $\hh$ and set
$$
\mathcal{F}(\hh) := \bigoplus_{n=0}^\infty \mathcal{F}^{(n)}(\hh) , \qquad \mathcal{F}^{(n)}(\hh) = \hh^{\otimes^n} ,
$$
where we have set $\hh^{\otimes 0} := \C$.
We introduce the vacuum vector $\Omega := (1,0,0,...) \in \FF(\hh)$.
The space $\mathcal{F}(\hh)$ is an inner product space
where the inner product is induced from the inner product in $\hh$. That is, on vectors $\eta_1 \otimes \cdots \eta_n, \varphi_1 \otimes \cdots \varphi_n \in \FF^{(n)}(\hh)$
we have
$$
\langle \eta_1 \otimes \cdots \eta_n , \varphi_1 \otimes \cdots \varphi_n \rangle := \prod_{i=1}^n
\langle \eta_{i} , \varphi_{i} \rangle_\hh .
$$
This definition extends to all of $\FF(\hh)$ by bilinearity and continuity. We introduce the bosonic Fock space
$$
\FF_s(\hh) := \bigoplus_{n=0}^\infty \mathcal{F}_s^{(n)}(\hh) , \qquad \mathcal{F}_s^{(n)}(\hh) := S_n \FF^{(n)}(\hh) ,
$$
where $S_n$ denotes the orthogonal projection onto the subspace of totally symmetric
tensors in $\FF^{(n)}(\hh)$. For $h \in \hh$ we introduce the so called creation operator $a^*(h)$ in $\FF_s(\hh)$
which is defined on vectors $\eta \in \FF^{(n)}_s(\hh)$
by
\beqn \label{eq:formala}
a^*(h) \eta := \sqrt{n+1} S_{n+1} ( h \otimes \eta ) \; .
\eeqn
The operator $a^*(h)$ extends by linearity to a densely defined linear operator on $\FF(\hh)$.
One can show that $a^*(h)$ is closable, c.f. \cite{reesim2}, and we denote its closure by the same symbol.
We introduce the annihilation operator by $a(h) := (a^*(h))^*$.
For a closed operator $A \in \hh$ with domain $D(A)$ we introduce the operator $\Gamma(A)$ and $d\Gamma(A)$ in $\FF(\hh)$ defined
on vectors $\eta = \eta_1 \otimes \cdots \otimes \eta_n \in \FF^{(n)}(\hh)$, with $\eta_i \in D(A)$, by
$$
\Gamma(A) \eta := A \eta_1 \otimes \cdots \otimes A \eta_n
$$
and
$$
d \Gamma(A) \eta := \sum_{i=1}^n \eta_1 \otimes \cdots \otimes \eta_{i-1} \otimes A \eta_i \otimes \eta_{i+1} \otimes \cdots \otimes \eta_n
$$
and extended by linearity to a densely defined linear operator on $\FF(\hh)$. One can show
that $d \Gamma(A)$ and $\Gamma(A)$ are closable, c.f. \cite{reesim2}, and we denote their closure by the same symbol.
The operators $\Gamma(A)$ and $d \Gamma(A)$ leave the subspace $\FF_s(\hh)$ invariant, that is,
their restriction to $\FF_s(\hh)$ is densely defined, closed, and has range contained in $\FF_s(\hh)$.
To define qed, we fix
$$
\hh := L^2( \R^3 \times \Z_2 )
$$
and set $\FF := \FF_s(\hh)$. We denote the norm of $\hh$ by $\| \cdot \|_\hh$.
We define the operator of the free field energy
by
$$
H_f := d \Gamma(M_\omega) ,
$$
where $\omega(k,\lambda) := \omega(k) := |k|$ and $M_\varphi$ denotes the operator of multiplication with the function $\varphi$.
For $f \in \hh$ we write
$$
a^*(f) = \sum_{\lambda=1,2} \int dk f(k,\lambda) a^*(k,\lambda) , \qquad a(f) = \sum_{\lambda=1,2} \int dk \overline{f(k,\lambda)} a^*(k,\lambda) .
$$
where $a(k,\lambda)$ and $a^*(k,\lambda)$ are operator-valued distributions.
They satisfy the
following commutation relations, which are to be understood in the sense of distributions,
$$
[a(k,\lambda), a^*(k',\lambda') ] = \delta_{\lambda \lambda'} \delta(k - k') , \qquad [a^\#(k,\lambda), a^\#(k',{\lambda'}) ] = 0 \; ,
$$
where $a^{\#}$ stands for $a$ or $a^*$. For $\lambda=1,2$ we introduce the so called polarization vectors
$$
\varepsilon(\cdot , \lambda) : S^2 := \{ k \in \R^3 | |k| = 1 \} \to \R^3
$$
to be measurable maps such that for each $k \in S^2$ the vectors $\varepsilon(k,1), \varepsilon(k,2),k$ form an orthonormal basis of $\R^3$.
We extend $\varepsilon(\cdot , \lambda)$ to $\R^3 \setminus \{ 0 \}$ by setting
$
\varepsilon(k,\lambda) := \varepsilon(k/|k|,\lambda)
$
for all nonzero $k$. For $x \in \R^3$ we define the field operator
\beqn \label{eq:afield}
A_{\sigma}(x) := \sum_{\lambda=1,2} \int \frac{dk \kappa_{\sigma,\Lambda}(k)}{\sqrt{2 |k|}} \left[ e^{-ik \cdot x}
\varepsilon(k,\lambda) a^*(k,\lambda) + e^{ik \cdot x} \varepsilon(k, \lambda) a(k, \lambda) \right] \ ,
\eeqn
where the function $\kappa_{\sigma,\Lambda}$ serves as a cutoff, which satisfies $\kappa_{\sigma,\Lambda}(k) = 1$ if $\sigma \leq |k| \leq \Lambda$ and which
is zero otherwise. $\Lambda > 0$ is an ultraviolet cutoff, which we assume to be fixed, and $\sigma \geq 0$ an infrared cutoff.
Next we introduce the atomic Hilbert space, which describes the configuration of $N$ electrons, by
$$
\HH_{\rm at} := \{ \psi \in L^2(\R^{3N}) | \psi(x_{\pi(1)},...,x_{\pi(N)}) = {{\rm sgn}(\pi)} \psi(x_1,...,x_N) , \pi \in \mathfrak{S}_N \} ,
$$
where $\mathfrak{S}_N$ denotes the group of permutations of $N$ elements, ${\rm sgn}$ denotes the signum of the permutation,
and $x_j \in \R^3$ denotes the coordinate of the $j$-th electron.
We will consider the following operator in $\HH := \HH_{\rm at} \otimes \FF$,
\begin{equation} \label{eq:hamiltoniandefinition}
H_{g,\beta,\sigma} := : \sum_{j=1}^N ( p_j + g A_{\sigma}(\beta x_j) )^2 : + V + H_f ,
\end{equation}
where
$p_j = - i \partial_{x_j}$,
$V = V(x_1,...,x_N)$ denotes the potential, and $:( \, \cdot \, ):$ stands for the Wick product.
The coupling constant $g \in \C$ is of  interest for the main result,  Theorem~\ref{thm:main1}.
The parameter $\beta \in \R$ will be used  in  Corollary~\ref{cor:main2}.
We will make the following assumptions on the potential $V$, which are related to the atomic Hamiltonian
$$
H_{\rm at} := - \Delta + V ,
$$
which acts in $\HH_{\rm at}$. We introduced the Laplacian $-\Delta := \sum_{j=1}^N p_j^2$.

\vspace{0.5cm}
\noindent
{\bf Hypothesis (H)} The potential $V$ satisfies the following properties:
\begin{itemize}
\item[(i)] $V$ is invariant under rotations and permutations, that is
\begin{align*}
&V(x_1,...,x_N) = V(R^{-1}x_1,...,R^{-1}x_N) , \quad \forall R \in SO(3)  , \\
&V(x_1,...,x_N) = V(x_{\pi (1)},..., x_{\pi( N)}) , \quad \forall \pi  \in \mathfrak{S}_N .
\end{align*}
\item[(ii)] $V$ is infinitesimally operator bounded with respect to $-\Delta$.
\item[(iii)] $E_{\rm at} := \inf \sigma(H_{\rm at} )$ is a non-degenerate isolated eigenvalue of $H_{\rm at}$.
\end{itemize}

\vspace{0.5cm}

Note that for the Hydrogen, $N=1$, the potential $V(x_1) = - {|x_1|}^{-1}$
satisfies Hypothesis (H). Moreover (ii) of Hypothesis (H) implies that
$H_{g,\beta,\sigma}$ is a self-adjoint operator with domain $D(-\Delta + H_f)$ and that $H_{g,\beta,\sigma}$ is essentially
self adjoint on any operator core for $-\Delta +H_f$,
 see for example \cite{H02,HH08}. For a precise definition of the operator in \eqref{eq:hamiltoniandefinition}, see Appendix A.
We will use the notation $D_r(w) := \{ z \in \C | |z - w| < r \}$ and $D_r := D_r(0)$. Let us now state the main result of the paper.

\begin{theorem} \label{thm:main1} Assume Hypothesis (H). Then there exists
a positive constant $g_0$ such that for all $g \in D_{g_0}$, $ \beta \in \R$, and $\sigma \geq 0$
the operator $H_{g,\beta,\sigma}$ has an eigenvalue $E_{\beta,\sigma}(g)$ with eigenvector $\psi_{\beta,\sigma}(g)$ and eigen-projection
$P_{\beta,\sigma}(g)$ satisfying the following properties.
\begin{itemize}
\item[(i)] for $g \in \R \cap D_{g_0}$, $E_{\beta,\sigma}(g) = {\rm inf} \sigma ( H_{g,\beta,\sigma}) $. 
\item[(ii)] $g \mapsto E_{\beta,\sigma}(g)$ and $g \mapsto \psi_{\beta,\sigma}(g)$
are analytic on $D_{g_0}$.
\item[(iii)] $g \mapsto P_{\beta,\sigma}(g)$ is analytic on $D_{g_0}$ and
$P_{\beta,\sigma}(g)^* = P_{\beta,\sigma}(\overline{g})$.
\end{itemize}
The functions $E_{\beta,\sigma}(g)$, $\psi_{\beta,\sigma}(g)$, and $P_{\beta,\sigma}(g)$ are
 bounded in $(g,\beta,\sigma) \in D_{g_0} \times \R \times [0,\infty)$ and depend continuously on $\sigma \geq 0 $.
\end{theorem}

The infrared cutoff $\sigma$ will be used in Sections \ref{sec:ana}
to relate the expansion coefficients to analytic perturbation theory.
We want to emphasize that the proof of Theorem \ref{thm:main1} does not use any form of gauge invariance. In particular
the conclusions hold if the terms quadratic in $A_{\sigma}$ are dropped from the Hamiltonian.
Using  Theorem \ref{thm:main1} and Cauchy's formula one can show the following corollary, see Section \ref{sec:prov}.

\begin{corollary} \label{cor:main1} Assume Hypothesis (H).
And let $g_0$,  $E_{\beta,\sigma}(g)$,  $\psi_{\beta,\sigma}(g)$ and
$P_{\beta,\sigma}(g)$ be given as in Theorem \ref{thm:main1}.
Then on $D_{g_0}$ we have the convergent power series expansions
\begin{equation} \label{eq:expansion1}
\psi_{\beta,\sigma}(g) = \sum_{n=0}^\infty \psi_{\beta,\sigma}^{(n)} g^n , \quad P_{\beta,\sigma}(g) = \sum_{n=0}^\infty P^{(n)}_{\beta,\sigma} g^n , \quad
E_{\beta,\sigma}(g) = \sum_{n=0}^\infty E^{(2n)}_{\beta,\sigma} g^{2n} ,
\end{equation}
where the coefficients satisfy the following properties:
$\psi_{\beta,\sigma}^{(n)}$, $E^{(n)}_{\beta,\sigma}$, and
$P^{(n)}_{\beta,\sigma}$
and depend continuously on $\sigma \geq 0$, and there exist finite constants $C_0, R$ such that for
all $n \in \N_0$ and $(\beta,\sigma) \in \R \times [0,\infty)$
we have
$\| \psi_{\beta,\sigma}^{(n)} \| \leq C_0 R^{n} $, $| E^{(2n)}_{\beta,\sigma} | \leq C_0 R^{2n} $, and
$\| P^{(n)}_{\beta,\sigma} \| \leq C_0 R^{n} $.
\end{corollary}

If we set $\beta = \alpha \geq 0$, $g = \alpha^{3/2}$, and $\sigma =0$, then we immediately obtain the following corollary.
It states that the ground state and the ground state energy of an atom in qed, in the scaling
limit where the ultraviolet cutoff is of the order of the Rydberg energy, admit
convergent expansions in the fine structure constant with uniformly bounded
coefficients.

\begin{corollary} \label{cor:main2} Assume Hypothesis (H). There exists a positive $\alpha_0$ and finite constants $C_0, R$ such that for $ 0 \leq \alpha \leq \alpha_0$ the operator
$H_{\alpha^{3/2},\alpha,0}$ has a ground state $\psi(\alpha^{1/2})$ with ground state energy
$E(\alpha)$ such that we have the convergent expansions
$$
\psi(\alpha^{1/2}) = \sum_{n=0}^\infty \psi_\alpha^{(n)} \alpha^{3n/2} \quad, \quad
E(\alpha) = \sum_{n=0}^\infty E^{(2n)}_\alpha \alpha^{3 n } ,
$$
and for all $n \in \N_0$ and $\alpha \geq 0$
we have $\|\psi_\alpha^{(n)} \| \leq C_0 R^{n} $ and $|E_\alpha^{(2n)}| \leq C_0 R^{2n} $.
\end{corollary}

Corollary \ref{cor:main2} improves the main theorem stated in
\cite{BFP09}. It provides a convergent expansion and
furthermore shows that the expansion coefficients are finite. Moreover, we show in the next section,
that the expansion coefficients $\psi_\alpha^{(n)}$ and $E_\alpha^{(2n)}$ can be calculated
using regular analytic perturbation theory. This yields a straightforward algorithm for
calculating the ground state and the ground state energy to arbitrary precision in $\alpha$.
We want to point out that the authors in \cite{BFP09} note that they could
 alternatively work with an ultraviolet cut-off of the order of
the rest energy of an electron, which, in the units used in this paper, corresponds to choosing
$\Lambda(\alpha) = \mathcal{O}(\alpha^{-2})$.
 The methods used in the proof of Theorem \ref{thm:main1}
 could also incorporate a certain $\alpha$ dependence of the cutoff. This would
lead to weaker conclusions, which are not only technical.

\section{Analytic Perturbation Theory}
\label{sec:ana}

In order to relate the expansions given in Theorem \ref{thm:main1} and Corollary \ref{cor:main2} to ordinary analytic perturbation theory, we introduce an infrared cutoff $\sigma > 0$.
In that case, analytic perturbation theory becomes applicable, and it is straightforward to show the following theorem. For completeness we provide a proof.

\begin{theorem} \label{thm:perturb1} Assume Hypothesis (H).
For $\sigma > 0$ and $\beta \in \R$, there is a positive $g_0$ such that for all
$g \in D_{g_0}$, the operator
$H_{g,\beta,\sigma}$ has a non-degenerate eigenvalue
$\widehat{E}_{\beta,\sigma}(g)$ with eigen-projection $\widehat{P}_{\beta,\sigma}(g)$ such that
the following holds.
\begin{itemize}
\item[(i)] For $g \in D_{g_0}$ we have $\widehat{E}_{\beta,\sigma}(g) = \inf \sigma(H_{g,\beta,\sigma})$, and $\widehat{E}_{\beta,\sigma}(0) = E_{\rm at}$.
\item[(ii)] $g \mapsto \widehat{E}_{\beta,\sigma}(g)$ and $g \mapsto \widehat{P}_{\beta,\sigma}(g)$ are analytic functions on $D_{g_0}$.
\item[(iii)] $\widehat{P}_{\beta,\sigma}(g)^* = \widehat{P}_{\beta,\sigma}(\overline{g})$ for all $g \in D_{g_0}$.
\end{itemize}
On $D_{g_0}$ we have convergent power series expansions
\begin{equation} \label{eq:perurbexpansion1}
\widehat{P}_{\beta,\sigma}(g) = \sum_{n=0}^\infty \widehat{P}_{\beta,\sigma}^{(n)} g^n , \quad
\widehat{E}_{\beta,\sigma}(g) = \sum_{n=0}^\infty \widehat{E}_{\beta,\sigma}^{(n)} g^n .
\end{equation}
\end{theorem}

\begin{proof} Fix $\sigma > 0$ and $\beta \in \R$. We introduce the subspaces
$\hh_\sigma^{(+)} := L^2(\{ k \in \R^3 | |k| \geq \sigma \} \times \Z_2)$
and $\hh_\sigma^{(-)} := L^2(\{ k \in \R^3 | |k| < \sigma \} \times \Z_2)$ of $\hh$, and we define the
associated Fock-spaces $\FF_{\sigma}^{(\pm)} := \FF_s(\hh_{\sigma}^{(\pm)})$. By $1_\sigma^{(\pm)}$ we denote
the identity operator in $\FF_{\sigma}^{(\pm)}$ and by $1_{\rm at}$ the identity operator in $\HH_{\rm at}$.
We consider the natural unitary isomorphism
$
U : \FF_{\sigma}^{(+)} \otimes \FF_{\sigma}^{(-)} \to \FF
$,
which is uniquely characterized by
$$
U( \{ h_1 \otimes_s \cdots \otimes_s h_n \} \otimes \{ g_1 \otimes_s \cdots \otimes_s g_m \} )
= h_1 \otimes_s \cdots \otimes_s h_n \otimes_s g_1 \otimes_s \cdots \otimes_s g_m ,
$$
for any $h_1,...,h_n \in \hh_{\sigma}^{(+)}$ and $g_1,...,g_m \in \hh_{\sigma}^{(-)}$.
We denote the trivial extension of $U$ to $\HH_{\rm at} \otimes \FF_{\sigma}^{(+)} \otimes \FF_{\sigma}^{(-)}$ by the same symbol.
We expand the Hamiltonian as follows. We write
$$
H_{g,\beta,\sigma} = H_0 + T_{\beta,\sigma}(g) ,
$$
with $H_0 := H_{\rm at} + H_f$ and
\begin{align*}
T_{\beta,\sigma}(g) &:= g \sum_{j=1}^N 2 p_j \cdot A_{\sigma}(\beta x_j) + g^2 : \sum_{j=1}^N A_{\sigma}(\beta x_j)^2 : .
\end{align*}
By $T_{\beta,\sigma}^{(+)}(g)$ we denote the unique operator in $\HH_{
\rm at} \otimes \FF_{\sigma}^{(+)}$ such that
$T_{\beta,\sigma}(g) = U ( T_{\beta,\sigma}^{(+)}(g) \otimes 1_{\sigma}^{(-)}) U^*$.
We have
$$
U^* H_{g,\beta,\sigma} U = \left( H_{0,\sigma}^{(+)} + T_{\beta,\sigma}^{(+)}(g)
 \right) \otimes 1_{\sigma}^{(-)} + 1_{\rm at} \otimes 1_{\sigma}^{(+)} \otimes H_{f,\sigma}^{(-)} \; ,
$$
where we introduced the following operators acting on the corresponding spaces
\begin{align*}
H_{0,\sigma}^{(+)} &= H_{\rm at} \otimes 1_{\sigma}^{(+)} + 1_{\rm at} \otimes H_{f,\sigma}^{(+)} \\
H_{f,\sigma}^{(-)} &= d \Gamma (M_{\chi_\sigma \omega}) , \quad
H_{f,\sigma}^{(+)} = d \Gamma (M_{(1-\chi_\sigma) \omega}) ,
\end{align*}
where $\chi_{\sigma}(k) = 1$ if $|k| < \sigma$ and zero otherwise.
Now observe that $H_{f,\sigma}^{(-)}$ has only one eigenvalue. That eigenvalue is 0, it is at the bottom
of the spectrum, it is non-degenerate and and its
eigenvector, $\Omega_\sigma^{(-)}$, is the vacuum of $\FF_\sigma^{(-)}$. This implies
that $H_{g,\beta,\sigma}$ and $H_{0,\sigma}^{(+)} + T_{\beta,\sigma}^{(+)}(g)$ have the same eigenvalues and the corresponding
eigen-spaces are in bijective correspondence.
Next observe that $H_{0,\sigma}^{(+)}$ has at the bottom of its spectrum an isolated non-degenerate eigenvalue which equals
$E_{\rm at}$.
Moreover, $g \mapsto H_{0,\sigma}^{(+)} + T_{\beta,\sigma}^{(+)}(g)$ is
an analytic family, since the interaction term is bounded with respect to $H_{0,\sigma}^{(+)}$.
Now by analytic perturbation theory, it follows that there exists an $\epsilon > 0$
such that for $g$ in a neighborhood of zero the following operator
is well defined
\begin{align} \label{eq:orginialprojection}
P_{\beta,\sigma}^{(+)}(g) &:= - \frac{1}{2 \pi i} \int_{|z -E_{\rm at}| = \epsilon} (H_{0,\sigma}^{(+)} +
 T_{\beta,\sigma}^{(+)}(g) - z )^{-1} dz .
\end{align}
Moreover, the operator $P_{\beta,\sigma}^{(+)}(g) $ projects
onto a one-dimensional space which is the eigen-space of $H_{0,\sigma}^{(+)} + T_{\beta,\sigma}^{(+)}(g) $ with eigenvalue
$\widehat{E}_{\beta,\sigma}(g)$. Furthermore, $P_{\beta,\sigma}^{(+)}(g) $ and $\widehat{E}_{\beta,\sigma}(g)$ depend analytically on $g$
and $\widehat{E}_{\beta,\sigma}(0) = E_{\rm at}$. We conclude that $\widehat{E}_{\beta,\sigma}(g)$ is a non-degenerate eigenvalue of
$H_{g,\beta,\sigma}$ with corresponding eigen-projection
\begin{equation} \label{eq:orginialprojection22}
\widehat{P}_{\beta,\sigma}(g) = U ( P_{\beta,\sigma}^{(+)}(g)\otimes P_{\Omega_\sigma^{(-)} } ) U^* ,
\end{equation}
and properties (i)-(iii) of the theorem are satisfied.
\end{proof}

We want to emphasize that the $g_0$ of Theorem \ref{thm:perturb1} depends on $\sigma$ and $\beta$ and we have
not ruled out that $g_0 \to 0$ as $\sigma \downarrow 0$. To rule this out we will need Theorem \ref{thm:main1}.
The expansion coefficients of the eigenvalue or the associated eigen-projection obtained on the one hand by renormalization,   \eqref{eq:expansion1},
and on the other
hand using analytic perturbation theory are equal.
To this end, note that for $\sigma > 0$ and $\beta \in \R$ there exists by Theorems \ref{thm:main1} and \ref{thm:perturb1}
 a ball $D_{r}$ of nonzero radius $r$,  such that the following holds.
The eigenvalue $\widehat{E}_{\beta,\sigma}(g)$ is non-degenerate for $g \in D_r$.
Thus $\widehat{E}_{\beta,\sigma}(g) = E_{\beta,\sigma}(g)$ on $D_r$ and hence
 $\widehat{P}_{\beta,\sigma}(g) = P_{\beta,\sigma}(g)$ on $D_r$. Thus the following remark is an immediate consequence of
 Theorems \ref{thm:main1} and \ref{thm:perturb1}.

\begin{remark} For all $\beta \in \R$ and $\sigma>0$ we have
$
{P}_{\beta,\sigma}^{(n)} = \widehat{P}_{\beta,\sigma}^{(n)}$ and ${E}_{\beta,\sigma}^{(n)} = \widehat{E}_{\beta,\sigma}^{(n)}$.
Moreover, $\widehat{P}_{\beta,\sigma}^{(n)}$ and $\widehat{E}_{\beta,\sigma}^{(n)}$ have a limit as $\sigma \downarrow 0$.
\end{remark}

Finally we want to note that $\widehat{P}^{(n)}_{\beta,\sigma}$ can be calculated, for example,
by first expanding the resolvent in Eq. \eqref{eq:orginialprojection} in powers of $g$ and then using
Eq. \eqref{eq:orginialprojection22}. This will then yield the coefficients $ \widehat{E}_{\beta,\sigma}^{(n)}$, for example by
expanding the right hand side of the identity
$$
\widehat{E}_{\beta,\sigma}(g) = \frac{ \langle \varphi_{\rm at} \otimes \Omega, H_{g,\beta,\sigma} \widehat{P}_{\beta,\sigma}(g) \varphi_{\rm at} \otimes \Omega \rangle}{
 \langle \varphi_{\rm at} \otimes \Omega, \widehat{P}_{\beta,\sigma}(g) \varphi_{\rm at} \otimes \Omega \rangle } ,
$$
where $\varphi_{\rm at}$ denotes the ground state of $H_{\rm at}$.

\section{Outline of the Proof}
\label{sec:outline}

The main method used in the proof of Theorem \ref{thm:main1} is
operator theoretic renormalization \cite{BFS98,BCFS03} and the fact that renormalization preserves
analyticity \cite{GH09,HH10}. The renormalization procedure is an iterated application of
the so called smooth Feshbach map. The smooth Feshbach map is reviewed in Appendix C and necessary properties of
it are summarized.
In this paper we will use many results stated in a previous paper \cite{HH10}.
Their generalization from the Fock space over $L^2(\R^3)$, as considered in \cite{HH10}, to a Fock space
over $L^2(\R^3 \times \Z_2)$ is straight forward. To be able to show that the renormalization transformation
is a suitable contraction we use that by rotation invariance
the renormalization procedure only involves kernels which do not contain any terms which are
linear in creation or annihilation operators. In section \ref{sec:symmetries} we define an $SO(3)$ action
on the atomic Hilbert space and the Fock space, which leaves the Hamiltonian invariant.
In section \ref{sec:ban} we introduce spaces which are needed
to define the renormalization transformation.
In section \ref{sec:ini} we show that after an initial Feshbach transformation the Feshbach map is
in a suitable Banach space. This allows us to use results of \cite{HH10} which are collected
in Section \ref{sec:ren:def}. In section \ref{sec:prov} we put all the pieces together and
prove Theorem \ref{thm:main1}. The proof is based on Theorems \ref{thm:inimain1} and \ref{thm:bcfsmain}.
In section \ref{sec:prov}, we also show Corollary \ref{cor:main1}.

We shall make repeated use of the so called pull-through formula which is given in Lemma \ref{lem:pullthrough}, in Appendix A.
Moreover we will use the notation that $\R_+ := [0,\infty)$.
Finally, let us note that using an appropriate
scaling we can assume without loss of generality that the distance between the lowest eigenvalue of $H_{\rm at}$ and the rest of the spectrum is one, i.e.,
\begin{equation} \label{eq:hatscale}
E_{\rm at,1} - E_{\rm at} = 1 ,
\end{equation}
where $E_{\rm at,1} := \inf ( \sigma( H_{\rm at}) \setminus\{ E_{\rm at}\}) $. Any Hamiltonian of the form
\eqref{eq:hamiltoniandefinition} satisfying Hypothesis (H) is up to a positive multiple unitarily equivalent to an operator
satisfying \eqref{eq:hatscale} and again Hypothesis (H), but with a rescaled potential and with different values for
$\sigma,\Lambda,\beta$, and $g$. More explicitly, with $\delta := E_{\rm at,1} - E_{\rm at}$ we have
\begin{equation} \label{eq:scaledhamiltonian}
\delta^{-1} S H_{g,\beta,\sigma} S^* =
\sum_{j=1}^N ( p_j + \widetilde{g} A_{\widetilde{\sigma},\widetilde{\Lambda}}(\widetilde{\beta} x_j ))^2 + V_\delta + H_f ,
\end{equation}
where $S$ is a the unitary transformation which leaves the vacuum invariant and
satisfies $S x_j S^* = \delta^{-1/2} x_j$ and $S a^\#(k) S^* = \delta^{-3/2} a^\#(\delta^{-1} k)$. We used the notation
$V_\delta := \delta^{-1} S V S^*$, $\widetilde{\beta} := \delta^{1/2} \beta$, $\widetilde{\Lambda} := \delta^{-1} \Lambda$, $\widetilde{\sigma} := \delta^{-1} \sigma$, and $\widetilde{g} = \delta^{1/2} g$. From the definition of $\delta$ it follows immediately from
\eqref{eq:scaledhamiltonian} that $\sum_{j=1}^N p_j^2 + V_\delta$ satisfies \eqref{eq:hatscale}.

\section{Symmetries}

\label{sec:symmetries}

Let us introduce the following canonical representation of $SO(3)$ on $\HH_{\rm at}$ and $\hh$. For $R \in SO(3)$ and $\psi \in \HH_{\rm at}$ we define
$$
\mathcal{U}_{\rm at}(R) \psi(x_1,...,x_N) = \psi(R^{-1} x_1, ... , R^{-1} x_N) .
$$
To define an $SO(3)$ representation on Fock space it is convenient to consider a different but equivalent representation of the
Hilbert space $\hh$.
We introduce the Hilbert space $\hh_0 := L^2(\R^3 ; \C^3)$. We consider the
subspace of transversal vector fields
$$
\hh_T := \{ f \in \hh_0 | k \cdot f(k) = 0 \} .
$$
It is straightforward to verify that the map
$\phi : \hh \to \hh_T $ defined by
\begin{eqnarray*}
 (\phi f)(k) := \sum_{\lambda=1,2} f(k,\lambda) \varepsilon(k,\lambda)
\end{eqnarray*}
establishes a unitary isomorphism with inverse
$$
(\phi^{-1} h)(k,\lambda) = h(k) \cdot \varepsilon(k,\lambda) .
$$
We define the action of $SO(3)$ on $\hh_T$ by
$$
(\mathcal{U}_T(R) h )(k) = R h(R^{-1} k) , \quad \forall h \in \hh_T , R \in SO(3) .
$$
The function $R \mapsto \phi^{-1} \mathcal{U}_T(R) \phi$ defines a representation of $SO(3)$ on $\hh$ which we denote
by  $\UU_\hh$.  For $R \in SO(3)$ and $f \in \hh$ it  is  given by
\begin{equation} \label{eq:repofh}
(\mathcal{U}_\hh(R) f )(k,\lambda) = \sum_{\widetilde{\lambda}=1,2}
D_{\lambda \widetilde{\lambda}}(R,k) f(R^{-1} k , \widetilde{\lambda}) ,
\end{equation}
where
$
D_{\lambda \widetilde{\lambda}}(R,k) := \varepsilon(k,\lambda) \cdot R \varepsilon(R^{-1}k, \widetilde{\lambda})
$.
This yields a representation on Fock space which we  denote by $\UU_\FF$. It is characterized by
\begin{equation} \label{eq:repofcreation2}
\UU_\FF(R) a^\#(f) \UU_\FF(R)^* = a^\#(\UU_\hh(R) f) \quad , \quad \UU_\FF(R) \Omega = \Omega .
\end{equation}
We have
\begin{equation} \label{eq:repofcreation}
\mathcal{U}_\FF(R) a^\#(k,\lambda) \UU_\FF(R)^* =
\sum_{\widetilde{\lambda}=1,2} D_{\widetilde{\lambda} \lambda}(R, Rk) a^\#(Rk, \widetilde{\lambda}) .
\end{equation}
We denote the representation on $\HH_{\rm at} \otimes \FF$ by
 $\mathcal{U} =  \mathcal{U}_{\rm at} \otimes \mathcal{U}_\FF$.
We have the following transformation properties of the operators $(x_j)_l$ and $(p_j)_l$, with $j=1,...,N$ and $l=1,2,3$,
\begin{align} \label{eq:transx}
\mathcal{U}(R) (x_j)_l \mathcal{U}(R)^* &= \sum_{m=1}^3 R_{ml} (x_j)_m = (R^{-1} x_j)_l  , \\
\mathcal{U}(R) (p_j)_l \mathcal{U}(R)^* &= \sum_{m=1}^3 R_{ml} (p_j)_m = (R^{-1} p_j)_l .
\end{align}
Moreover, the transformation property of the $l$-th component of the field operator $A_{\sigma,l}(x_j)$ is
\begin{equation} \label{eq:transA}
\mathcal{U}(R) A_{\sigma,l}(x_j) \mathcal{U}(R)^* = \sum_{m=1}^3 R_{ml} A_{\sigma,m}(x_j) = (R^{-1} A)_l( x_j) .
\end{equation}
This can be seen as follows. For fixed $x \in \R^3$ and $l=1,2,3$ define the function
\begin{equation} \label{defofintkernela}
f_{(l,x)}(k,\lambda) := \frac{\kappa_{\sigma,\Lambda}(k)}{\sqrt{ 2 |k|}} \varepsilon(k,\lambda)_l e^{- i k \cdot x } .
\end{equation}
Eq. \eqref{eq:transA} follows since by \eqref{eq:repofh} we have  $\UU_\hh(R) f_{(l,x)} = \sum_{m=1}^3 R_{ml} f_{(m, R x)} $.
We call a linear  operator $A$ in the Hilbert space  $\HH$ rotation invariant if
$A = \UU(R) A \UU(R)^*$ for all $R \in SO(3)$ and likewise for operators in $\FF$ and $\HH_{\rm at}$.
From \eqref{eq:transx}--\eqref{eq:transA} it is evident to see  that the Hamiltonian $H_{g,\beta,\sigma}$ defined in
 \eqref{eq:hamiltoniandefinition},   is rotation invariant.

\begin{lemma} \label{lem:mainidea} Let $f \in \hh$.  If   $a^\#(f)$ is an operator which is invariant under rotations,  then $f=0$.
\end{lemma}

\begin{proof} Invariance implies
$$
a^\#(f) = \UU_\FF(R) a^\#(f) \UU_\FF(R)^* = a^\#(\UU_\hh(R) f)
$$
and therefore $\UU_\hh(R) f = f$. This implies that for  $\widehat{f} := \phi f$ we have
\begin{equation} \label{eq:fRk}
\widehat{f}(Rk) = R \widehat{f}(k) .
\end{equation}
Let  $H_l$ denote  the space of spherical harmonics of angular momentum $l$.
We note that $L^2(\R^3;\C^3) = \bigoplus_{l=0}^\infty L^2(\R^+)\bigotimes H_l\bigotimes\C^3$ where each summand is invariant under the representation of $SO(3), f(\cdot) \mapsto Rf(R^{-1}\cdot)$.
It follows that $\widehat{f} = \bigoplus_{l=0}^\infty\widehat {f}_l$ where each $\widehat {f}_l$ is invariant.  By Fubini's theorem there is a null set $\Lambda_1 \subset \R^+$ such that for a countable dense set $\mathcal C$ of $R \in SO(3)$ there is a null set $\Lambda_2(t) \subset S^2$ so that $R\widehat {f}_l(t,R^{-1}e) = \widehat {f}_l(t,e)$ for all $t$ in the complement of $\Lambda_1$, $R \in \mathcal C$, and $e$ in the complement of $\Lambda_2(t)$.  But since $H_l$ is just the space of spherical harmonics of angular momentum $l$, $\widehat {f}_l(t,e)$ is continuous in the variable $e$ so we can take $\mathcal C = SO(3)$ and $\Lambda_2(t) = \emptyset$.

In particular if  $R e_3 = e_3$, then $\widehat{f_l}(t, e_3) = R \widehat{f_l}(t, e_3)$. This implies that  $\widehat{f_l}(t, e_3) = c_l(t) e_3$ for some function $c_l$ on $[0,\infty)\setminus\Lambda_1$.  Rotating $e_3$ into an arbitrary $e \in S^2$ and using the invariance we find $\widehat{f_l}(t, e) = c_l(t)e$
which in turn implies that  $\widehat{f}(k) = c(|k|) k$ almost everywhere. But a function of this type is an element of $\hh_T$ only if it is 0.
\end{proof}

\section{Banach Spaces of Hamiltonians}
\label{sec:ban}

In this section we introduce Banach spaces of integral kernels, which
parameterize certain subspaces of the space of bounded operators on Fock space.
These subspaces are suitable to study an iterative application of the Feshbach map
and to formulate the contraction property. We mainly follow the exposition in \cite{BCFS03}.
However, we use a different class of Banach spaces.

The renormalization transformation will be defined on operators acting on the reduced Fock space
$\mathcal{H}_{\rm red}:= P_{\rm red} \FF$,
where we introduced the notation $P_{\rm red}:= \chi_{[0,1]}(H_f)$.
We will investigate bounded operators in $\mathcal{B}(\mathcal{H}_{\rm red})$ of the form
\beqn \label{eq:sum}
H(w) := \sum_{m+n \geq 0} H_{m,n}(w) ,
\eeqn
with
\begin{align}
& H_{m,n}(w) := H_{m,n}(w_{m,n}) , \nonumber \\
& H_{m,n}(w_{m,n}) := P_{\rm red} \int_{\un{B}_1^{m+n}} \frac{ d \mu( {K}^{(m,n)})}{|{K}^{(m,n)}|^{1/2}} a^*({K}^{(m)}) w_{m,n}(H_f, {K}^{(m,n)}) a(\widetilde{{K}}^{(n)}) P_{\rm red} , \quad m+n \geq 1 , \label{eq:defhmn11} \\
& H_{0,0}(w_{0,0}) := w_{0,0}(H_f) , \nonumber
\end{align}
where $w_{m,n} \in L^\infty([0,1] \times \un{B}_1^m \times \un{B}_1^n)$
is an integral kernel for $m+n \geq 1$, $w_{0,0} \in L^\infty([0,1])$, and $w$ denotes the sequence of
integral kernels $(w_{m,n})_{m,n \in \N_0^2}$.
We have used and will henceforth use the following notation. We set $K = (k, \lambda ) \in \R^3 \times \Z_2$, and write
\begin{align*}
& \un{X} := X \times \Z_2 \quad , \quad B_1 := \{ x \in \R^3 | |x|< 1 \} \\
& K^{(m)} := ({K}_1, ... ,{K}_m ) \in \left( {\R}^{3} \times \Z_2 \right)^m ,
\quad \widetilde{{K}}^{(n)} := (\widetilde{{K}}_1, ... , \widetilde{{K}}_n ) \in \left( {\R}^{3} \times \Z_2 \right)^n , \\
& {K}^{(m,n)} := ({K}^{(m)}, \widetilde{{K}}^{(n)}) \\
& \int_{\un{X}^{m+n}} d {K}^{(m,n)} := \int_{{X}^{m+n}} \sum_{(\lambda_1,...,\lambda_m,\widetilde{\lambda}_1,...,\widetilde{\lambda}_n) \in \Z_2^{m+n} } dk^{(m)} d\widetilde{k}^{(n)} \\
& dk^{(m)} := \prod_{i=1}^m {d^3 k_i} , \quad d\widetilde{k}^{(n)} := \prod_{j=1}^n {d^3 \widetilde{k}_j} , \quad
d K^{(m)} := d K^{(m,0)} , \quad d \widetilde{K}^{(n)} := d K^{(0,n)} , \\
& d \mu (K^{(m,n)}) := (8 \pi )^{-\frac{{m+n}}{2}} d K^{(m,n)} \\
& a^*({K}^{(m)}) := \prod_{i=1}^m a^*({K}_i) , \quad a(\widetilde{{K}}^{(m)}) := \prod_{j=1}^m a(\widetilde{{K}}_j) \\
& | {K}^{(m,n)}| := | {K}^{(m)} | \cdot | \widetilde{{K}}^{(n)}| , \quad | {K}^{(m)} | := |k_1| \cdots |k_m | , \quad | \widetilde{{K}}^{(m)} | := |\widetilde{k}_1| \cdots |\widetilde{k}_m | , \\
& \Sigma[{K}^{(m)}] := \sum_{i=1}^n |k_m | \; .
\end{align*}
Note that in view of the pull-through formula \eqref{eq:defhmn11} is equal to
\beqn \label{eq:defintegralkernel}
\int_{\underline{B}_1^{m+n}} \frac{ d \mu(K^{(m,n)})}{|K^{(m,n)}|^{1/2}} a^*(K^{(m)}) \chi(H_f + \Sigma[K^{(m)}] \leq 1 ) w_{m,n}(H_f , K^{(m,n)})
\chi(H_f + \Sigma[\tilde{K}^{(n)}] \leq 1) a(\tilde{K}^{(n)} ) \; .
\eeqn
Thus we can restrict attention to integral kernels $w_{m,n}$ which are essentially supported on the sets
\begin{eqnarray*}
\underline{Q}_{m,n} &:=& \{ ( r , K^{(m,n)}) \in [0,1] \times \underline{B}_1^{m+n} \ | \ r \leq 1 -
\max(\Sigma[K^{(m)}],
\Sigma[\widetilde{K}^{(m)}]) \} , \quad m + n \geq 1 .
\end{eqnarray*}
Moreover, note that integral kernels can always be assumed to be symmetric. That is, they lie in the range of the symmetrization operator,
which is defined as follows,
\begin{eqnarray} \label{eq:symmetrization}
w_{M,N}^{({\rm sym})}(r , K^{(M,N)}) := \frac{1}{N!M!} \sum_{\pi \in S_M} \sum_{\widetilde{\pi} \in S_N} {w}_{M,N}(r,
K_{\pi(1)},\ldots,K_{\pi(N)}, \widetilde{K}_{\widetilde{\pi}(1)},\ldots,\widetilde{K}_{\widetilde{\pi}(M)}).
\end{eqnarray}
Note that \eqref{eq:defhmn11} is understood in the sense of forms. It defines a densely defined form
which can be seen to be bounded using the expression \eqref {eq:defintegralkernel} and Lemma \ref{kernelopestimate}.
Thus it uniquely determines a bounded operator which we denote by $H_{m,n}(w_{m,n})$. This is explained in more
detail in Appendix A. We have the following lemma.
\begin{lemma} For $w_{m,n} \in L^\infty([0,1] \times \un{B}_1^m \times \un{B}_1^n)$\label{lem:operatornormestimates} we have
\beqn \label{eq:operatornormestimate1}
\|H_{m,n}(w_{m,n}) \|_{} \leq \| w_{m,n} \|_{\infty} ( n! m!)^{-1/2} \; .
\eeqn
\end{lemma}
The proof follows using Lemma \ref{kernelopestimate}
and the estimate
\begin{equation} \label{eq:intofwKminus2}
 \int_{\underline{S}_{m,n}} \frac{d K^{(m,n)}}{|K^{(m,n)}|^{2}} \leq \frac{(8 \pi)^{m+n}}{{n! m!}} ,
\end{equation}
where $\underline{S}_{m,n} := \{ (K^{(m)},\widetilde{K}^{(n)}) \in \underline{B}_1^{m+n}
 \ | \Sigma[K^{(m)}] \leq 1 , \Sigma[\widetilde{K}^{(n)}] \leq 1 \}$.
The renormalization procedure will involve kernels which lie in the following Banach spaces.
We shall identify the space $L^\infty(\underline{B}_1^{m+n}; C[0,1])$ with a subspace of $L^\infty([0,1]\times \underline{B}_1^{m+n})$ by
setting $$w_{m,n}(r,K^{(m,n)}) = w_{m,n}(K^{(m,n)})(r)$$ for $w_{m,n} \in L^\infty(\underline{B}_1^{m+n}; C[0,1])$. For example in  (i) and (ii)
of  Definition \ref{def:wgartenhaag} we use this identification.
The norm in $L^\infty(\underline{B}_1^{m+n}; C[0,1])$
is given by
$$
\| w_{m,n} \|_{\underline{\infty}} := \esssup_{K^{(m,n)} \in \underline{B}_1^{m+n}} {\rm sup}_{r \geq 0}| w_{m,n}(K^{(m,n)})(r)| .
$$
We note that for $w \in L^\infty(\underline{B}_1^{m+n}; C[0,1])$ we have $ \| w \|_\infty \leq \| w \|_{\underline{\infty}}$.
Conditions (i) and (ii) of the
following definition are needed for the injectivity property stated in Theorem \ref{thm:injective}, below.

\begin{definition} \label{def:wgartenhaag}
We define $\WW_{m,n}^\#$ to be the Banach space consisting of functions $w_{m,n} \in
L^\infty(\underline{B}_1^{m+n};C^1[0,1])$ satisfying the following properties:
\begin{itemize}
\item[(i)] $ w_{m,n} (1 - \chi_{\underline{Q}_{m,n}} ) = 0$, for $m + n \geq 1$,
\item[(ii)] $w_{m,n}(\cdot,K^{(m)}, \widetilde{K}^{(n)})$ is totally
symmetric in the variables $K^{(m)}$ and $\widetilde{K}^{(n)}$
\item[(iii)] the following norm is finite
$$
\| w_{m,n} \|^\# := \| w_{m,n} \|_{\underline{\infty}} + \| \partial_r w_{m,n} \|_{\underline{\infty}} .
$$
\end{itemize}
Hence for almost all $K^{(m,n)} \in \underline{B}_1^{m+n}$ we have $w_{m,n}(\cdot,K^{(m,n)}) \in C^1[0,1]$, where
the derivative is denoted by $\partial_r w_{m,n}$.
For $0<\xi < 1$, we define the Banach space
$$
 \mathcal{W}^\#_{\xi} := \bigoplus_{(m,n) \in \N_0^2 } \mathcal{W}_{m,n}^\# \
$$
to consist of all sequences $w =( w_{m,n})_{m,n \in \N_0}$ satisfying
$$
\| w \|_\xi^\# := \sum_{(m,n)\in \N_0^2} \xi^{-(m+n)} \| w_{m,n}\|^\# < \infty .
$$
\end{definition}

\begin{remark}
{\em
We shall also use the norm $\| w_{m,n} \|^\#$ for any integral kernel $w_{m,n} \in L^\infty( \underline{B}_1^{m+n}; C^1[0,1])$.
 Note that $\| w_{m,n}^{({\rm sym})} \|^\# \leq \| w_{m,n} \|^\#$.}
\end{remark}

Given $w \in \mathcal{W}_\xi^\#$, we
write $w_{\geq r}$ for the vector in $\mathcal{W}_\xi^\#$ given by
$$
(w_{\geq r})_{m + n} = \left\{ \begin{array}{ll} w_{m,n} & , \quad {\rm if} \ m+n \geq r \\ 0 & , \quad {\rm otherwise} . \end{array} \right.
$$
We will use the following balls to define the renormalization transformation
\begin{align*}
\mathcal{B}^\#(\alpha,\beta,\gamma) := \left\{ w \in \mathcal{W}_\xi^\# \left| \| \partial_r w_{0,0} - 1 \|_\infty \leq \alpha , \
|w_{0,0}(0) | \leq \beta
, \ \| w_{\geq 1} \|_{\xi}^\# \leq \gamma \right. \right\} .
\end{align*}
For $w \in \mathcal{W}^\#_{\xi}$, it is easy to see using \eqref{eq:operatornormestimate1} that
$
H(w) := \sum_{m,n} H_{m,n}(w)
$
converges in operator norm with bounds
\begin{align} \label{eq:opestimatgeq12}
& \| H(w) \|_{} \leq \| w\|_\xi^\# , \\
 \label{eq:opestimatgeq123}
& \| H(w_{\geq r} ) \|_{} \leq \xi^r \| w_{\geq r} \|_\xi^\# .
\end{align}
We shall use the notation
$$
W[w] := \sum_{m+n \geq 1} H_{m,n}(w) .
$$
We will use the following theorem, which is a straightforward generalization of a theorem proven in \cite{BCFS03}. A proof
can also be found in \cite{HH10}.

\begin{theorem} \label{thm:injective} The map $H : \WW_\xi^\# \to \mathcal{B}(\HH_{\rm red})$ is injective and bounded.
\end{theorem}

\begin{definition}
Let $\WW_\xi$ denote the Banach space consisting of strongly analytic functions on $D_{1/2}$ with values
in $\WW_\xi^\#$ and norm given by
$$
\| w  \|_\xi := \sup_{z \in D_{1/2}} \| w(z) \|_\xi^\# .
$$
\end{definition}
For $w \in \WW_\xi$ we will use the notation $w_{m,n}(z, \cdot) := (w_{m,n}(z))(\cdot)$.
We extend the definition of $H(\cdot)$ to $\WW_\xi$ in the natural way: for $w \in \WW_\xi$, we set
$$
\left( H(w) \right) (z) := H(w(z))
$$
and likewise for $H_{m,n}(\cdot)$ and $W[\cdot]$.
We say that a kernel $w \in \WW_\xi$ is symmetric if $w_{m,n}(\overline{z}) = \overline{w_{n,m}(z)}$  for all $z \in D_{1/2}$. Note that because of Theorem
\ref{thm:injective} we have the following lemma.
\begin{lemma} \label{lem:symmetry} Let $ w \in \WW_\xi$. Then
 $w$ is symmetric if and only if $H(w(\overline{z}))= H(w(z))^*$ for all $z \in D_{1/2}$.
\end{lemma}
The renormalization transformation will be defined on the following balls in $\mathcal{W}_\xi$
\begin{eqnarray*}
\lefteqn{ \mathcal{B}(\alpha,\beta,\gamma) } \\
&& := \left\{ w \in \mathcal{W}_\xi \left| \sup_{z \in D_{1/2}} \| \partial_r w_{0,0}(z) - 1 \|_\infty \leq \alpha , \,
\sup_{z \in D_{1/2}} | w_{0,0}(z,0) + z | \leq \beta
, \, \| w_{\geq 1} \|_{\xi} \leq \gamma \right. \right\} .
\end{eqnarray*}
We define on the space of kernels $\WW_{m,n}^\#$ a
natural representation of $SO(3)$, $\mathcal{U}_\WW$, which by Theorem \ref{thm:injective}
is uniquely determined by
\begin{equation} \label{eq:rinvkerop}
H (\mathcal{U}_\WW(R) w_{m,n}) = \mathcal{U}_\FF(R) H(w_{m,n}) \mathcal{U}_\FF^*(R) , \quad \forall R \in SO(3) ,
\end{equation}
and it is given by $\mathcal{U}_\WW(R) w_{0,0}(r) = w_{0,0}(r)$ and for $m+n \geq 1$ by
\begin{eqnarray}
\lefteqn{
\left( \mathcal{U}_\WW(R)w_{m,n} \right)(r , k_1,\lambda_1,\ldots,\widetilde{k}_n,\widetilde{\lambda}_n) } \label{eq:defofronkernels} \\
&&= \sum_{(\lambda_1',..., \widetilde{\lambda}_n') \in \Z_2^{m+n}}
     D_{\lambda_1 \lambda_1'}(R,k_1) \cdots D_{\widetilde{\lambda}_n \widetilde{\lambda}_n'}(R,\widetilde{k}_n)
     w_{m,n}(r , R^{-1} k_1,\lambda_1',\ldots ,R^{-1} \widetilde{k}_n, \widetilde{\lambda}_n') . \nonumber
\end{eqnarray}
  That \eqref{eq:defofronkernels} implies \eqref{eq:rinvkerop} can be seen from \eqref{eq:repofcreation}.
 The representation on $\WW_{m,n}^\#$ yields a natural representation
on $\WW_\xi^\#$, which is given  by $(\mathcal{U}_\WW(R) w )_{m,n} = \mathcal{U}_\WW(R) w_{m,n}$ for all $R \in SO(3)$.
We say that a kernel $w_{m,n} \in \WW_{m,n}^\# $ is rotation invariant if $\mathcal{U}_\WW(R) w_{m,n} = w_{m,n}$ for all $R \in SO(3)$
and we say
a kernel $w \in \WW_{\xi}^\# $ is rotation invariant if each component is rotation invariant.

\begin{lemma} \label{lem:wHinvequiv} (i)
Let $w_{m,n} \in \WW_{m,n}^\#$. Then $H(w_{m,n})$ is rotation invariant if and only if $w_{m,n}$ is rotation invariant.
  Let $w \in \WW_{\xi}^\#$. Then $H(w)$ is rotation invariant if and only if $w$ is rotation invariant.
(ii) If
$w_{m,n} \in \WW_{m,n}^\#$ with $m+n=1$  is rotation invariant, then $w_{m,n} = 0$.
\end{lemma}
\begin{proof} (i).
The if part follows from \eqref{eq:rinvkerop}. The only if part follows from \eqref{eq:rinvkerop} and the injectivity
of the map $H(\cdot)$, see Theorem \ref{thm:injective}. (ii) Let $w_{1,0} \in \WW_{1,0}^\#$ be rotation invariant. Then $w_r$ defined
by $w_r(k,\lambda) := w_{1,0}(r,k,\lambda)$ is in $\hh$ for all $r \in [0,1]$.
By \eqref{eq:defofronkernels},   \eqref{eq:repofh}, and
\eqref{eq:repofcreation2} it follows that $a^*(w_r)$ is rotation invariant. By Lemma \ref{lem:mainidea},
$w_r = 0$. The proof of the corresponding statement for  $\WW_{0,1}^\#$ is analogous.
\end{proof}

To state that the contraction property of the renormalization transformation will need to introduce the balls of integral kernels
which are invariant under rotations
$$
\mathcal{B}_0(\alpha,\beta,\gamma) := \{ w \in \mathcal{B}(\alpha,\beta,\gamma) | \
w_{m,n}(z) \ {\rm is \ rotation \ invariant \ for \ all \ }  z \in D_{1/2} \ \} .
$$

To show the continuity of the ground state and the ground state energy as a function of the infrared cutoff we need to introduce a coarser norm
in $\WW_{m,n}^\#$. The supremum norm is to fine. To this end we introduce the Banach space $L^2_\omega(\underline{B}_1^{m+n};C[0,1])$  with norm
$$
\| w_{m,n} \|_{2} := \left[ \int_{\underline{B}_1^{m+n}} \frac{d  K^{(m,n)}}{(8\pi)^{m+n} |K^{(m,n)}|^2} \sup_{r \in [0,1]} | w_{m,n}(r,K^{(m,n)}) |^2 \right]^{1/2} .
$$
Observe that $L^\infty( \underline{B}_1^{m +n}; C[0,1]) \subset L^2_\omega(\underline{B}_1^{m+n};C[0,1])$ and that by \eqref{eq:intofwKminus2} we have
\begin{equation} \label{eq:infinity2ineq}
\| w_{m,n} \|_{2} \leq \frac{ \|w_{m,n} \|_{\infty} }{\sqrt{ n! m!}} , 
\end{equation}
for all
$w_{m,n} \in \WW_{m,n}^\#$.
We have the following lemma which is a consequence of  Lemma \ref{kernelopestimate}.
\begin{lemma} For $w_{m,n} \in L^2_{\omega}( \underline{B}_1^{m+n}; C[0,1])$\label{lem:operatornormestimates2}   we have
\beqn \label{eq:operatornormestimate12}
\|H_{m,n}(w_{m,n}) \|_{} \leq \| w_{m,n} \|_{2} \; .
\eeqn
\end{lemma}
\begin{definition}
Let $S$ be topological space. We say that the mapping $w : S \to \WW_\xi^\#$ is componentwise $L^2$--continuous (c-continuous) if for all
$m,n \in \N_0$ the map $s \mapsto w_{m,n}(s)$ is a $L^2_{\omega}( \underline{B}_1^{m+n}; C[0,1])$--valued continuous function, that is
$$
\lim_{s \in S, s \to s_0} \left\| w(s_0)_{m,n} - w(s)_{m,n} \right\|_2 = 0
$$
for all $s_0 \in S$.
\end{definition}
The above notion of continuity for integral kernels, yields continuity of the associated operators with respect to the operator norm topology. This is
the content of the following lemma.
\begin{lemma} \label{lem:pcontop}
Let $w: S \to \WW_\xi^\#$ be c-continuous and uniformly bounded, that is \\ $\sup_{s \in S} \| w(s) \|^\#_\xi < \infty$. Then
$H(w(\cdot)) : S \to \mathcal{B}(\HH_{\rm red})$ is continuous, with respect to the operator norm topology.
\end{lemma}
\begin{proof}
From Lemma \ref{lem:operatornormestimates2}
it follows that $H_{m,n}(w(s)) \stackrel{\| \cdot \|_{}}{\longrightarrow} H_{m,n}(w(s_0))$ as $s$ tends to $s_0$.
The lemma now follows from a simple argument using the estimate \eqref{eq:opestimatgeq123}
 and the uniform bound on $w(\cdot)$.
\end{proof}

\section{Initial Feshbach Transformations}

\label{sec:ini}

In this section we shall assume that the assumptions of Hypothesis (H) hold. Without loss of generality, see Section \ref{sec:outline}, we
assume that the distance between the lowest eigenvalue of $H_{\rm at}$ and the rest of the spectrum is one, that is
\begin{equation} 
 \inf \left( \sigma( H_{\rm at}) \setminus \{ E_{\rm at} \} \right) - E_{\rm at} = 1 .
\end{equation}
Let $\chi_1$ and $\chib_1$ be two functions in $C^\infty(\R_+;[0,1])$ with $\chi_1^2 + \chib_1^2 = 1$,
$\chi_1 = 1$ on $[0,3/4)$, and $\supp \chi_1 \subset [0,1]$. For an explicit choice of $\chi_1$ and $\chib_1$ see for example
\cite{BCFS03}.
We use the abbreviation $\chi_1 = \chi_1(H_f)$ and $\chib_1 = \chib_1(H_f)$.
It should be clear from the context whether $\chi_1$ or $\chib_1$ denotes a function or an operator.
By $\varphi_{\rm at}$ we denote the normalized eigenstate of $H_{\rm at}$ with eigenvalue
 $E_{\rm at}$ and by $P_{\rm at}$ the eigen-projection of $H_{\rm at}$ corresponding to the
 eigenvalue $E_{\rm at}$.
 By Hypothesis (H) the range of
$P_{\rm at}$ is one dimensional. This allows us to identify the range of
$P_{\rm at} \otimes P_{\rm red}$ with $\HH_{\rm red}$, and we will do so.
We define $\chi^{(I)}(r) := P_{\rm at} \otimes \chi_1(r)$ and
$\chib^{(I)}(r) = \bar{P}_{\rm at} \otimes 1 + P_{\rm at} \otimes \chib_1(r)$,
with $\bar{P}_{\rm at} = 1 - P_{\rm at}$. We set
$\chi^{(I)} := \chi^{(I)}(H_f)$ and $\chib^{(I)} := \chib^{(I)}(H_f)$. It is
evident to see that ${\chi^{(I)}}^2 + {\chib^{(I)}}^2 = 1$. The next theorem is the main
theorem of this section. It states properties about the Feshbach map and the associated auxiliary operator, see Appendix C.

\begin{theorem} \label{thm:inimain1} Assume Hypothesis (H). For any $0 < \xi < 1$ and any positive numbers $\delta_1,\delta_2,\delta_3$ there exists a positive number
$g_0$ such that following is satisfied. For all
$(g,\beta,\sigma,z) \in D_{g_0} \times \R \times \R_+ \times D_{1/2}$
the pair of operators
$(H_{g,\beta,\sigma} - z - E_{\rm at}, H_0 - z - E_{\rm at} )$ is a Feshbach pair for $\chi^{(I)}$.
The operator valued function
\begin{equation} \label{eq:qdef111}
 Q_{\chi^{(I)}}(g,\beta,\sigma,z) := Q_{\chi^{(I)}}( H_{g,\beta,\sigma} - z - E_{\rm at}, H_0 - z -E_{\rm at} )
\end{equation}
defined  on $ D_{g_0} \times \R \times \R_+ \times D_{1/2}$ is bounded, analytic in
 $(g,z)$, and a continuous function of  $(\sigma,z)$.
There exists a unique kernel
$w^{(0)}(g,\beta,\sigma,z) \in \WW_\xi^\#$ such that
\begin{equation} \label{eq:inimainA}
H(w^{(0)}(g,\beta,\sigma,z)) \cong F_{\chi^{(I)}}( H_{g,\beta,\sigma} - z - E_{\rm at} , H_0 - z - E_{\rm at} ) \upharpoonright \ran P_{\rm at} \otimes P_{\rm red} .
\end{equation}
Moreover, $w^{(0)}$ satisfies the following properties.
\begin{itemize}
\item[(a)] We have
$w^{(0)}(g,\beta,\sigma) := w^{(0)}(g,\beta,\sigma, \cdot ) \in \mathcal{B}_0(\delta_1,\delta_2,\delta_3)$
for all $(g,\beta,\sigma) \in D_{g_0} \times \R \times \R_+$.
\item[(b)]
$w^{(0)}(g,\beta,\sigma)$ is a symmetric kernel for all $(g,\beta,\sigma) \in ( D_{g_0} \cap \R) \times \R \times \R_+$.
\item[(c)]
The function $(g,z) \mapsto w^{(0)}(g,\beta,\sigma,z)$ is a $\WW_\xi^\#$-valued analytic function on $D_{g_0} \times D_{1/2}$ for
all $(\beta,\sigma) \in \R \times \R_+$.
\item[(d)] The function $(\sigma,z) \mapsto w^{(0)}(g,\beta,\sigma,z) \in \WW_\xi^\#$
is a c-continuous function on $\R_+ \times D_{1/2}$ for all $(g,\beta) \in D_{g_0} \times \R$.
\end{itemize}
\end{theorem}
The remaining part of this section is devoted to the proof of Theorem \ref{thm:inimain1}.
Throughout this section we assume that
\begin{equation} \label{eq:zzetarelation}
z = \zeta - E_{\rm at } \in D_{1/2} .
\end{equation}
To prove Theorem \eqref{thm:inimain1}, we write the interaction part of the Hamiltonian in terms
of integral kernels as follows,
$$
H_{g,\beta,\sigma} = H_{\rm at} + H_f + : W_{g, \beta,\sigma} : ,
$$
\begin{equation} \label{eq:sumofws}
W_{g,\beta,\sigma} := \sum_{m+n=1,2} W_{m,n}(g,\beta,\sigma) .
\end{equation}
where $W_{m,n}(g,\beta,\sigma) := \underline{H}_{m,n}(w_{m,n}^{(I)}(g,\beta,\sigma))$
with
\begin{align}
 \label{eq:defhlinemn}
& \underline{H}_{m,n}(w_{m,n}) := \int_{{(\underline{\R}^3)}^{m+n}} \frac{ dK^{(m,n)}}{|K^{(m,n)}|^{1/2}}
a^*(K^{(m)}) w_{m,n}(K^{(m,n)}) a(\widetilde{K}^{(n)}) ,
\end{align}
and
\begin{align}
&w^{(I)}_{1,0}(g,\beta,\sigma)( K) := 2 g \sum_{j=1}^N p_j \cdot \varepsilon(k,\lambda) \frac{\kappa_{\sigma,\Lambda}(k)e^{ i \beta k \cdot x_j }}{\sqrt{2}} , \label{defofwI} \\
&w^{(I)}_{1,1}(g,\beta,\sigma)(K,\widetilde{K}) := g^2 \sum_{j=1}^N \varepsilon(k,\lambda) \cdot \varepsilon(\widetilde{k},\widetilde{\lambda}) \frac{\kappa_{\sigma,\Lambda}(k)e^{- i \beta k \cdot x_j }}{\sqrt{2}} \frac{\kappa_{\sigma,\Lambda}(\widetilde{k})e^{ i \beta \widetilde{k}\cdot x_j }}{\sqrt{2}} ,
\nonumber \\
&w^{(I)}_{2,0}(g,\beta,\sigma)( K_1, K_2 ) := g^2 \sum_{j=1}^N \varepsilon(k_1,\lambda_1) \cdot \varepsilon({k}_2,{\lambda}_2) \frac{\kappa_{\sigma,\Lambda}(k_1)e^{ - i \beta k_1 \cdot x_j }}{\sqrt{2}} \frac{\kappa_{\sigma,\Lambda}(k_2)e^{ - i \beta k_2 \cdot x_j }}{\sqrt{2}} ,
\nonumber
\end{align}
$w^{(I)}_{0,1}(g,\beta,\sigma)(\widetilde{K}) := w^{(I)}_{0,1}(\overline{g},\beta,\sigma)(\widetilde{K})^*$, and $w^{(I)}_{0,2}(g,\beta,\sigma)(\widetilde{K}_1,\widetilde{K}_2) := \overline{w^{(I)}_{2,0}(\overline{g},\beta,\sigma)(\widetilde{K}_1,\widetilde{K}_2 )}$.
We note that \eqref{eq:defhlinemn} is understood in the sense of forms, c.f. Appendix A.
We set
\begin{align*}
w^{(I)}_{0,0}(z)(r) := H_{\rm at} - z + r .
\end{align*}
By $w^{(I)}$ we denote the vector consisting of the components $w^{(I)}_{m,n}$ with $m+n=0,1,2$.

The next theorem establishes the Feshbach property. 
To state it, we denote by ${P}_0$ the orthogonal projection onto the closure of ${\ran \chib^{(I)}}$.
We will use the convention that
$( H_0 - z)^{-1} \chib^{(I)}$
stands for $( H_0 - z \upharpoonright \ran \chib^{(I)} ) )^{-1} \chib^{(I)}$, and that
$( H_0 - z)^{-1} P_0$ stands for $( H_0 - z \upharpoonright \ran {P}_0 )^{-1} P_0 $.
The proof of the Feshbach property is based on the fact that
\begin{equation} \label{ini:eq1}
{\rm inf} \sigma ( H_0 \upharpoonright \ran {P_0} ) = E_{\rm at} + \sfrac{3}{4} ,
\end{equation}
which follows directly from the definition, and the
fact that the interaction part of the Hamiltonian is bounded with respect to the free Hamiltonian.
\begin{theorem} \label{ini:thm1} Let $| E_{\rm at} - \zeta | < \frac{1}{2}$. Then
\begin{equation} \label{ini:thm1:eq1}
\left\| ( ( H_0 - \zeta ) \upharpoonright \ran P_0 )^{-1} \right\| \leq 4 .
\end{equation}
There is a $C<\infty$ and $g_0 > 0$ such that for all $(\beta,\sigma) \in \R \times \R_+$ and $|g| < g_0$,
\begin{equation} \label{ini:thm1:eq2}
\left\| ( H_0 - \zeta)^{-1} \chib^{(I)} W_{g,\beta,\sigma} \right\| \leq C |g| , \quad \left\| W_{g, \beta,\sigma} ( H_0 - \zeta)^{-1} \chib^{(I)} \right\| \leq C |g| ,
\end{equation}
and $(H_{g,\beta,\sigma}-\zeta,H_0-\zeta)$ is a Feshbach pair for $\chi^{(I)}$. The function
$(g,\beta,\sigma, \zeta) \mapsto  ( H_0 - \zeta)^{-1} \chib^{(I)} W_{g,\beta,\sigma}$ on $\C \times \R \times \R_+ \times  D_{1/2}(E_{\rm at})$
is analytic in  $(g,\zeta)$ and continuous in
$(\sigma,\zeta)$.
\end{theorem}

\begin{proof} Eq. \eqref{ini:thm1:eq1} follows directly from Eq. \eqref{ini:eq1}. We will only show the
first inequality of \eqref{ini:thm1:eq2}, since the second one will then follow from
$$
\| W_{g, \beta,\sigma} ( H_0 - \zeta)^{-1} \chib^{(I)} \| = \| ( H_0 - \overline{\zeta})^{-1} \chib^{(I)} W_{\overline{g}, \beta,\sigma} \| ,$$
where we used that the norm of an operator is equal to the norm of its adjoint.
 The Feshbach property will follow by Lemma \ref{fesh:thm2} as a consequence of \eqref{ini:thm1:eq1} and \eqref{ini:thm1:eq2}.
For $| E_{\rm at} - \zeta | < \frac{1}{2}$, we estimate
\begin{eqnarray}
\left\| ( H_0 - \zeta )^{-1} \chib^{(I)} W_{g, \beta,\sigma} \right\|
&\leq & \left\| ( H_0 - \zeta)^{-1} P_0 (H_0 - E_{\rm at} + 2 ) P_0 (H_0 - E_{\rm at} + 2 )^{-1} W_{g, \beta,\sigma} \right\| \nonumber \\
&\leq & \left\| \frac{H_0 - E_{\rm at} + 2 }{ H_0 - \zeta } P_0 \right\| \left\| ( H_0 - E_{\rm at} + 2 )^{-1} W_{g, \beta,\sigma} \right\| .
\label{eq:FestimateAA}
\end{eqnarray}
Using the spectral theorem we estimate the first factor in \eqref{eq:FestimateAA} by
\begin{equation} \label{eq:FestimateI}
\left\| \frac{H_0 - E_{\rm at} + 2 }{ H_0 - \zeta } P_0 \right\| \leq
\sup_{r \geq 0} \left| \frac{ \frac{3}{4} + 2 + r }{E_{\rm at} + \frac{3}{4} - \zeta + r } \right|
\leq \sup_{r \geq 0} \left| \frac{ 11 + 4 r}{1 + 4r } \right| \leq 11 .
\end{equation}
It remains to estimate the second factor in \eqref{eq:FestimateAA}. We insert \eqref{eq:sumofws}
and use the triangle
inequality,
\begin{equation} \label{eq:estingthesummands}
 \left\| ( H_0 - E_{\rm at} + 2 )^{-1} W_{g, \beta,\sigma} \right\| \leq \sum_{m+n=1,2}
\left\| ( H_0 - E_{\rm at} + 2 )^{-1} W_{m,n}(g, \beta,\sigma) \right\| .
\end{equation}
We estimate each summand occurring in the sum on the right hand side individually.
To estimate the summands with $m+n=2$ we first use the trivial bound
\begin{eqnarray} \label{eq:estonhfw}
\left\| ( H_0 - E_{\rm at} + 2 )^{-1} W_{m,n}({g, \beta,\sigma}) \right\| \leq \left\| ( H_f + 1 )^{-1} W_{m,n}({g, \beta,\sigma}) \right\| .
\end{eqnarray}
The right hand side of \eqref{eq:estonhfw} is estimated for $(m,n)=(0,2)$ as follows,
\begin{eqnarray} \lefteqn{
 \left\| ( H_f + 1 )^{-1} W_{0,2}({g, \beta,\sigma}) \right\|} \nonumber \\
 && \leq \frac{|g|^2 N}{2} \left[ \int_{({\underline{\R}^3})^2 } \frac{ d
\widetilde{K}^{(2)}}{|\widetilde{K}^{(2)}|^2} \left| \kappa_{\sigma,\Lambda}(\widetilde{k}_1) \right|^2
 \left| \kappa_{\sigma,\Lambda}(\widetilde{k}_2) \right|^2 \sup_{r \geq 0} \frac{ ( r + |\widetilde{k}_1| + |\widetilde{k}_2|)^2}{(r + 1)^2} \right]^{1/2}
\nonumber \\
 && \leq \frac{|g|^2 N}{2} \left[ 3 \| \kappa_{\sigma,\Lambda}/ \omega \|_\hh^4 + 6 \| \kappa_{\sigma,\Lambda}/ \omega \|_\hh^2
 \| \kappa_{\sigma,\Lambda} \|_\hh^2 \right]^{1/2} , \label{eq:someestimate}
\end{eqnarray}
where in the first inequality we used Lemma \ref{kernelopestimate} and
 in the last inequality we used the following estimate for $r \geq 0$,
$$\frac{ ( r + |\widetilde{k}_1| + |\widetilde{k}_2|)^2}{(r + 1)^{2}} \leq 3 ( 1 + |\widetilde{k}_1|^2 + |\widetilde{k}_2|^2) .
$$
To estimate the
right hand side of \eqref{eq:estonhfw} for $(m,n)=(2,0)$ we
use the  fact that the norm of an operator is equal to the norm of its adjoint, the pull-through formula, and a similar estimate
as used in \eqref{eq:someestimate},
$$
 \left\| ( H_f + 1 )^{-1} W_{2,0}({g, \beta,\sigma}) \right\| =
\left\| W_{0,2}({\overline{g}, \beta,\sigma}) ( H_f + 1 )^{-1} \right\| \leq {\rm r.h.s.} \ \eqref{eq:someestimate} .
$$
 To estimate the
right hand side of \eqref{eq:estonhfw} for $(m,n)=(1,1)$ we first use the pull-through formula and then
 Lemma \ref{kernelopestimate} to obtain
\begin{eqnarray}
\lefteqn{ \left\| ( H_f + 1 )^{-1} W_{1,1}({g, \beta,\sigma}) \right\|} \nonumber \\
&& \leq \frac{|g|^2 N}{2} \left[ \int_{{(\underline{\R}^3)}^2 } \frac{ d K^{(1,1)}}{|K^{(1,1)}|^2} \left| \kappa_{\sigma,\Lambda}({k}_1) \right|
 \left| \kappa_{\sigma,\Lambda}(\widetilde{k}_1) \right| \sup_{r \geq 0} \frac{ ( r + |{k}_1|)( r + |\widetilde{k}_1|)}{(r + 1)^2}\right]^{1/2} \nonumber \\
 && \leq \frac{|g|^2 N}{2} \left[ 2 \| \kappa_{\sigma,\Lambda}/ \omega \|_\hh^4 + 2 \| \kappa_{\sigma,\Lambda}/ \omega \|_\hh^2 \| \kappa_{\sigma,\Lambda} \|_\hh^2 \right]^{1/2} , \label{eq:usedforqcont1}
\end{eqnarray}
where in the last inequality we used the following estimate for $r \geq 0$,
$$
\frac{ ( r + |{k}_1|)( r + |\widetilde{k}_1|)}{(r + 1)^2} \leq 2 + |k_1|^2 + |\widetilde{k}_1|^2 .
$$
To estimate the summands with $m+n=1$ on the right hand side of \eqref{eq:estingthesummands}
we insert the trivial identity $1=(H_f + 1)^{1/2} (-\Delta + 1)^{1/2} (H_f + 1)^{-1/2} (-\Delta + 1)^{- 1/2}$ and obtain
the estimate
\begin{eqnarray*}
\lefteqn{ \left\| ( H_0 - E_{\rm at} + 2)^{-1} W_{m,n}({g, \beta,\sigma}) \right\|} \\
 && \leq \left\| \frac{ ( H_f + 1 )^{1/2} (H_{\rm at} - E_{\rm at} + 1 )^{1/2}}{ H_0 - E_{\rm at} + 2 } \right\|
  \left\| ( H_{\rm at} - E_{\rm at} + 1 )^{-1/2} (-\Delta + 1)^{1/2} \right\| \\
 && \times
  \left\| (-\Delta + 1 )^{-1/2} (H_f + 1 )^{-1/2} W_{m,n}({g, \beta,\sigma}) \right\| .
\end{eqnarray*}
The first factor on the right hand side is bounded by $1/2$, which follows from a trivial application of the spectral theorem.
The second factor on the right hand side is bounded, since $V$ is infinitesimally operator bounded with respect to $-\Delta$. The last factor
on the right hand side is estimated as follows. For $m+n=1$,
\begin{eqnarray}
\lefteqn{ \left\| (-\Delta + 1 )^{-1/2} (H_f + 1 )^{-1/2} W_{m,n}({g, \beta,\sigma}) \right\| } \nonumber \\
 &&\leq 2 |g| \sum_{j=1}^N \sum_{l=1}^3
\left\| \frac{ (p_j)_l }{ (-\Delta + 1 )^{1/2}} \right\| \left\| (H_f + 1 )^{-1/2} \left[ \delta_{m0} \underline{H}_{1,0}( \omega^{1/2}f_{(l,\beta x_j)}) +
  \delta_{n0} \underline{H}_{0,1}( \omega^{1/2}\overline{f_{(l,\beta x_j)}}) \right] \right\| \nonumber \\
&& \leq 6 N |g| \left( \| \kappa_{\sigma,\Lambda} /\omega \|_\hh^2 + \delta_{n0} \| \kappa_{\sigma,\Lambda} / \sqrt{\omega}\|_\hh^2 \right)^{1/2} , \label{eq:usedforqcont2}
\end{eqnarray}
where in the first inequality we used the triangle inequality and  \eqref{defofintkernela},
and in the second inequality we used the pull-through formula and Lemma \ref{kernelopestimate}.
Collecting estimates we obtain the desired bound on the second factor in \eqref{eq:FestimateAA}.
The statement about the analyticity and continuity follow from the explicit expression and the bounds
 in \eqref{eq:FestimateAA}--\eqref{eq:usedforqcont2}.
\end{proof}

As a consequence of  the first equation in \eqref{ini:thm1:eq2} it follows that the operator valued function
\eqref{eq:qdef111} is uniformly bounded for $g_0$ sufficiently small.
Theorem  \ref{ini:thm1} furthermore implies that \eqref{eq:qdef111} is continuous in $(\sigma,z)$ and analytic in $(g,z)$,
provided $g_0$ is sufficiently small.
Next we want to show that there exists a $w^{(0)}(g,\beta,\sigma,z) \in \WW_\xi^\#$ such that \eqref{eq:inimainA} holds.
Uniqueness will follow from Theorem \ref{thm:injective}.
In view of Theorem \ref{ini:thm1} we can define for $z = \zeta - E_{\rm at} \in D_{1/2}$ and $g$ sufficiently small
the Feshbach map and express it in terms of a Neumann series.
\begin{eqnarray*}
\lefteqn{ F_{\chi^{(I)}}( H_{g,\beta,\sigma} - \zeta, H_0 - \zeta) \upharpoonright X_{\rm at} \otimes \HH_{\rm red} } \\
&& = \left( T + \chi^{} W \chi^{} - \chi^{} W \chib^{} ( T + \chib^{} W_{}\chib^{} )^{-1}
\chib^{} W_{} \chi^{} \right) \upharpoonright X_{\rm at} \otimes \HH_{\rm red} \\
&& = \left( T^{} + \chi W^{} \chi - \chi W^{} \chib \sum_{n=0}^\infty \left( - {T^{}}^{-1} \chib W^{} \chib \right)^n
{T^{}}^{-1} \chib W^{} \chi \right) \upharpoonright X_{\rm at} \otimes \HH_{\rm red} \; ,
\end{eqnarray*}
where here we used the abbreviations
$T^{} = H_0 - \zeta$, $W^{} = W^{}_{g,\beta,\sigma}$, $\chi = \chi^{(I)}$, $\chib = \chib^{(I)}$.
We normal order above expression, using the pull-through formula. To this end we use the identity of
Theorem \ref{thm:wicktheorem}, see Appendix B.
Moreover we will use the definition
\begin{eqnarray} \label{eq:defofWW}
\underline{W}_{p,q}^{m,n}[w](K^{(m,n)}) &:=& \int_{{(\underline{\R}^3)}^{p+q}} \frac{d X^{(p,q)}}{|X^{(p,q)}|^{1/2}} a^*(X^{(p)}) w_{m+p,n+q}(K^{(m)}, X^{(p)}, \widetilde{K}^{(n)},
 \widetilde{X}^{(q)}) a(\widetilde{X}^{(q)}) . \nonumber
\end{eqnarray}
We obtain a sequence of integral kernels
$\widetilde{w}^{(0)}$, which are given as follows.
For $M+N \geq 1$,
\begin{eqnarray} \label{eq:defofwmnschlange}
\lefteqn{ \widetilde{w}^{(0)}_{M,N}(g,\beta,\sigma,z)(r , K^{(M,N)})} \label{initial:eq7} \\ &&=
( 8 \pi )^{\frac{M+N}{2}} \sum_{L=1}^\infty (-1)^{L+1}
\sum_{\substack{ (\umm,\upp,\unn,\uqq) \in \N_0^{4L}: \\ |\umm|=M, |\unn|=N, \\ 1 \leq m_l+p_l+q_l+n_l \leq 2 } } \prod_{l=1}^L \left\{
\binom{ m_l + p_l}{ p_l} \binom{ n_l + q_l}{ q_l }
\right\} \nonumber \\ && \times
V_{(\umm,\upp,\unn,\uqq)}[w^{I}(g,\beta,\sigma,\zeta)](r,K^{(M,N)}). \nonumber
\end{eqnarray}
 Furthermore,
\begin{align*}
\widetilde{w}^{(0)}_{0,0}(g,\beta,\sigma,z)(r) = - z + r + \sum_{L=2}^\infty (-1)^{L+1}
\sum_{(\upp,\uqq)\in \N_0^{2L}: p_l+q_l = 1, 2}
V_{(\uzz,\upp,\uzz,\uqq)}[w^{(I)}(g,\beta,\sigma,\zeta)](r) \; .
\end{align*}
Above we have used the definition
\begin{eqnarray} \label{eq:defofV}
\lefteqn{ V_{\umm,\upp,\unn,\uqq}[w](r, K^{(|\umm|,|\unn|)}) := } \\
&&
\left\langle \varphi_{\rm at} \otimes \Omega , F_0[w](H_f + r) \prod_{l=1}^L \left\{
\underline{W}_{p_l,q_l}^{m_l,n_l}[w](K^{(m_l,n_l)})
 F_l[w](H_f + r + \widetilde{r}_l ) \right\} \varphi_{\rm at} \otimes \Omega
\right\rangle \nonumber ,
\end{eqnarray}
where for $l=0,L$ we set $F_l[w](r) := \chi_1(r )$ , and for $l=1,...,L-1$ we set
\begin{eqnarray*}
F_l[w](r) := F[w](r) := \frac{ { \chib}^{(I)}(r )^2}{ w_{0,0}(r )} .
\end{eqnarray*}
Moreover, see \eqref{eq:rltildedef} for the definition of $\widetilde{r}_l$.
We define
$w^{(0)}(g,\beta,\sigma,z) := \left( \widetilde{w}^{(0)} \right)^{({\rm sym})}(g,\beta,\sigma,z) $.
So far we have determined $w^{(0)}$ on a formal
level. We have not yet shown that the involved series converge.
Our next goal is to show estimates \eqref{eq:feb2:1}, \eqref{eq:feb2:2}, and \eqref{eq:feb2:3},
below. These estimates will then imply that
 $w^{(0)}(g,\beta,\sigma,z) \in \WW^\#_\xi$ and they will be used to show part (a) of Theorem \ref{thm:inimain1}.
To this end we need an estimate on $V_{\umm,\upp,\unn,\uqq}[w^{(I)}]$, which is given in the following lemma.

\begin{lemma} \label{initial:thmE22} There exists finite constants $C_W$ and $C_F$
 such that with $C_W(g) := C_{W} |g|$ we have for $|\zeta-E_{\rm at} | < {1/2} $,
\begin{align} \label{initial:thmmain:eq2}
 \| V_{\umm,\upp,\unn,\uqq}[w^{(I)}(g,\beta,\sigma,\zeta)] \|^\# \leq (L + 1) C_F^{L+1} C_W(g)^L ,
\end{align}
for all $(g,\beta,\sigma) \in \C \times \R \times \R_+$.
\end{lemma}
To show this lemma we will use the estimates from the following lemma
and we introduce the following operator
$$
G_0 := - \Delta + H_f + 1 .
$$
\begin{lemma} \label{ini:lemelemestimates} There exist finite constants $C_W$ and $C_F$ such that the following holds.
We have
\begin{align}
& \| G_0^{-1/2} \underline{W}_{p,q}^{m,n}[w^{(I)}(g,\beta,\sigma,\zeta)](K^{(m,n)}) G_0^{-1/2} \| \leq C_{W} g^{m+p+n+q} , \label{eq1:ini:lemelemestimates}
\end{align}
for all $(g,\beta,\sigma,\zeta, K^{(m,n)}) \in \C \times \R \times \R_+ \times \C \times \underline{B}_1^{m+n}$.
For $ |\zeta- E_{\rm at} | < {1/2}$, we have
\begin{align}
& \| G_0^{1/2} F[w^{(I)}( g,\beta,\sigma,\zeta ) ](r + H_f ) G_0^{1/2} \| \leq C_F \label{eq2:ini:lemelemestimates} , \\
& \| G_0^{1/2} \partial_r F[w^{(I)}( g,\beta,\sigma,\zeta)](r + H_f ) G_0^{1/2} \| \leq C_F \label{eq3:ini:lemelemestimates} ,
\end{align}
for all $(g,\beta,\sigma,r) \in \C \times \R \times \R_+ \times \R_+$.
\end{lemma}
\begin{proof}
First we show \eqref{eq1:ini:lemelemestimates}.
For simplicity we drop the $(g, \beta, \sigma, \zeta)$--dependence in the notation.
If $p=q=0$ it follows directly from the definition that
$$
{\rm l . h. s. \ of \ } \eqref{eq1:ini:lemelemestimates} \ \leq 2 |g|^{m+n+p+q} N .
$$
To see the corresponding estimate for $ p + q \geq 1$ we first introduce the notation
\begin{equation} \label{eq:defofbor}
B_0(r) := (- \Delta + r + 1 )^{-1/2}.
\end{equation}
Hence
by definition $B_0(H_f) = G_0^{-1/2}$.
Using the pull-through formula and Lemma \ref{kernelopestimate} we see that
\begin{eqnarray}
\lefteqn{
I^{m,n}_{p,q} := \left\| G_0^{-1/2} \underline{W}_{p,q}^{m,n}[w^{(I)}](K^{(m,n)}) G_0^{-1/2} \right\| } \nonumber \\
 && \leq \int_{{(\underline{\R}^3)}^{p+q}} \frac{d X^{(p,q)}}{|X^{(p,q)}|^2}
\sup_{r \geq 0} \Bigg[ \left\| B_0( r + \Sigma[X^{(p)}]) w_{m+p,n+q}^{(I)}( K^{(m)}, X^{(p)} , \widetilde{K}^{(n)},
 \widetilde{X}^{(q)}) B_0( r + \Sigma[\widetilde{X}^{(q)}]) \right\|^2
 \nonumber \\
 && \times \left( r + \Sigma[{X}^{(p)}]\right)^p \left( r + \Sigma[\widetilde{X}^{(q)}]\right)^q \Bigg] ,
 \label{eq:contsigmainitial1}
\end{eqnarray}
where we used the trivial estimate for $r \geq 0$,
\begin{equation} \label{eq:trivestimateforkernel}
\prod_{l=1}^p \left( r + \Sigma[{X}^{(l)}]\right) \leq \left( r + \Sigma[{X}^{(p)}]\right)^p .
\end{equation}
Now we use \eqref{eq:contsigmainitial1} to estimate the remaining cases for $m , n,p,q$ separately. We find
$$
I^{m,n}_{p,q} \leq
\left\{ \begin{array}{l} |g| 2 N \| {\kappa_{\sigma,\Lambda}}/\omega \|_\hh , \qquad {\rm if } \ \ S =1,\ p+q =1 , \\
  |g|^2 N \| {\kappa_{\sigma,\Lambda}}/\omega \|_\hh^{p+q} , \qquad {\rm if } \ \ S = 2 ,\ \max(p,q) = 1 , \\
 |g|^2 N \left( \|  {\kappa_{\sigma,\Lambda}}/\omega \|_\hh^2 + 2 \| \kappa_{\sigma,\Lambda} / \omega \|_\hh^2 \| \kappa_{\sigma,\Lambda} / \omega^{1/2} \|_\hh \right)^{1/2}
  , \qquad {\rm if } \ \ S = 2 , \ \max(p,q) = 2 ,
\end{array} \right.
$$
with $S := m + n + p + q$.
Collecting estimates, \eqref{eq1:ini:lemelemestimates} follows.
Next we show \eqref{eq2:ini:lemelemestimates}. Inserting two times the identity
$1 = (H_0 + r - E_{\rm at} + 1)^{1/2} (H_0 + r - E_{\rm at} + 1)^{-1/2}$ into the left hand side of
\eqref{eq2:ini:lemelemestimates} we find,
\begin{eqnarray*}
   {\rm l . h. s. \ of \ } \eqref{eq2:ini:lemelemestimates} \leq \left\| G_0^{1/2} (H_0 + r - E_{\rm at} + 1)^{-1/2} \right\|^2 \left\| \frac{H_{0} + r - E_{\rm at} + 1}{H_{0} + r - \zeta }
 \left[ \chib^{(I)}(H_f + r ) \right]^2 \right\| .
\end{eqnarray*}
The first factor is bounded since $V$ is infinitesimally bounded with respect to
$-\Delta$. The second factor
can be  bounded using a similar estimate as
\eqref{eq:FestimateI}. Finally \eqref{eq3:ini:lemelemestimates} is estimated in a similar way using
$$
F[w^{(I)}(g,\beta,\sigma,\zeta)]'(r) = \frac{- \left[ \chib^{(I)}(r) \right]^2}{\left(w_{0,0}^{(I)}(\zeta)(r)\right)^2} +
\frac{2 \chib^{(I)}(r) \partial_r \chib^{(I)}(r)}{w_{0,0}^{(I)}(\zeta)(r)}
$$
and the bound
\begin{eqnarray*}
\left\| \frac{ H_0 + r - E_{\rm at} + 1}{( H_0 + r -\zeta)^2} \left[ \chib^{(I)}(H_f + r ) \right]^2 \right\| &&\leq
\left\| \frac{ H_0 + r - E_{\rm at} + 1}{(H_0 + r - E_{\rm at} - 1/2)^2} \left[ \chib^{(I)}(H_f + r ) \right]^2 \right\| \\
&&\leq \sup_{r \geq 0} \left| \frac{ r + \frac{3}{4} + 1}{(r + 1/4)^2} \right| \leq 32 .
\end{eqnarray*}
\end{proof}

\noindent {\it Proof of Lemma \ref{initial:thmE22}.}
We estimate $\| V_{\umm,\upp,\unn,\uqq}[w^{(I)}(g,\beta,\sigma,\zeta)] \|_{\underline{\infty}}$ using
\begin{equation} \label{eq:babyestimate1}
| \langle \varphi_{\rm at} \otimes \Omega , A_1 A_2 \cdots A_n \varphi_{\rm at} \otimes \Omega \rangle | \leq \| A_1 \|_{\rm op} \| A_2 \|_{\rm op} \cdots \| A_n \|_{\rm op},
\end{equation}
where $\| \cdot \|_{\rm op}$ denotes the operator norm,
and Inequalities \eqref{eq1:ini:lemelemestimates} and \eqref{eq2:ini:lemelemestimates}.
To estimate $\| \partial_r V_{\umm,\upp,\unn,\uqq}[w^{(I)}(g,\beta,\sigma,\zeta)] \|_{\underline{\infty}}$ we first calculate
the derivative using the Leibniz rule. The resulting expression is estimated using again \eqref{eq:babyestimate1}
and Inequalities \eqref{eq1:ini:lemelemestimates}--\eqref{eq3:ini:lemelemestimates}.
\qed

\vspace{0.5cm}

Now we are ready to establish Inequalities \eqref{eq:feb2:1}--\eqref{eq:feb2:3}, below.
Recall that we assume \eqref{eq:zzetarelation}.
Let $S^L_{M,N}$ denote the set of tuples $(\umm,\upp,\unn,\uqq) \in \N_0^{4L}$ with
$|\umm|=M$, $|\unn|=N$, and $1 \leq m_l+p_l+q_l+n_l \leq 2$.
We estimate the norm of \eqref{eq:defofwmnschlange} using \eqref{initial:thmmain:eq2} and
find, with $\widetilde{\xi} := (8 \pi)^{-1/2} \xi$,
\begin{align}
\| w^{(0)}_{\geq 1}(g,\beta,\sigma,z) \|_\xi^\# &= \sum_{M+N \geq 1} {\xi}^{-(M+N)} \| \widetilde{w}_{M,N}(g,\beta,\sigma,z) \|^\# \nonumber \\
&\leq \sum_{M+N\geq 1} \sum_{L=1}^\infty \sum_{(\umm,\upp,\unn,\uqq) \in S^L_{M,N}}
\widetilde{\xi}^{-(M+N)} 4^L \| V_{\umm,\upp,\unn,\uqq}[w^{(I)}(g,\beta,\sigma,\zeta)] \|^\# \nonumber \\
&\leq \sum_{L=1}^\infty \sum_{M+N\geq 1} \sum_{(\umm,\upp,\unn,\uqq) \in S^L_{M,N} }
\widetilde{\xi}^{-|\umm|-|\unn|} (L+1) C_F \left( 4 C_W(g) C_F \right)^L \nonumber \\
&\leq\sum_{L=1}^\infty (L+1) 14^L \widetilde{\xi}^{-2L} C_F \left( 4 C_W(g) C_F \right)^L \; , \label{eq:feb2:1}
\end{align}
for all $(g,\beta,\sigma) \in \C \times \R \times \R_+$,
where in the second line we used $\binom{ m + p }{ p} \leq 2^{m+p}$
and in the last line we used  $|\umm|+|\unn| \leq 2L$ and that the number of elements  $(\umm,\upp,\unn,\uqq) \in \N_0^L$ with
$1 \leq m_l+n_l+p_l+q_l \leq 2$ is bounded by $14^{L}$.
A similar but simpler estimate yields
\begin{align}
\sup_{r \in [0,1]} | \partial_r w^{(0)}_{0,0}(g,\beta,\sigma,z)(r) - 1 |
&\leq \sum_{L=2}^\infty \sum_{(\upp,\uqq) \in \N_0^{2L}: p_l+q_l= 1, 2}
 \| V_{\uzz,\upp,\uzz,\uqq}[w^{(I)}(g,\beta,\sigma,\zeta)] \|^\# \nonumber \\
&\leq \sum_{L=2}^\infty 3^L (L+1) C_F \left( C_W(g) C_F \right)^L \; , \label{eq:feb2:2}
\end{align}
for all $(g,\beta,\sigma) \in \C \times \R \times \R_+$.
Analogously we have for all $(g,\beta,\sigma) \in \C \times \R \times \R_+$,
\begin{align}
| w^{(0)}_{0,0}(g,\beta,\sigma,z)(0) + z |
& \leq \sum_{L=2}^\infty \sum_{(\upp,\uqq) \in \N_0^{2L}: p_l+q_l= 1, 2}
\| V_{\uzz,\upp,\uzz,\uqq}[w^{(I)}(g,\beta,\sigma,\zeta)] \|^\# \nonumber
\\ \label{eq:feb2:3}
& \leq \sum_{L=2}^\infty 3^L (L+1) C_F \left( C_W(g) C_F \right)^L .
\end{align}
In view of the definition of $C_W(g)$ the right hand sides in \eqref{eq:feb2:1}--\eqref{eq:feb2:3} can be made arbitrarily small for sufficiently small $|g|$.
This implies that the kernel $w^{(0)}(g,\beta,\sigma,z)$ is in $\WW_\xi^\#$ and that the inequalities in the definition of $\mathcal{B}_0(\delta_1,\delta_2,\delta_3)$
are satisfied. Rotation invariance of $w^{(0)}$ follows since the right hand side of \eqref{eq:inimainA} is invariant under rotations and Lemma \ref{lem:wHinvequiv}.
(b) follows from the properties of the right hand side of \eqref{eq:inimainA} and Lemma \ref{lem:symmetry}. It remains to show (c) and (d).
(c) respectively (d) follow from the convergence established in \eqref{eq:feb2:1}--\eqref{eq:feb2:3}, which is
uniform in $(g,\beta,\sigma,z) \in D_{g_0} \times \R \times \R_+ \times D_{1/2}$, and Lemma \ref{lem:analytwI}
respectively Lemma \ref{lem:contwI}, shown below.

\begin{lemma} \label{lem:analytwI} The mapping $(g,z) \mapsto V_{\umm,\upp,\unn,\uqq}[w^{(I)}(g,\beta,\sigma,E_{\rm at} + z )]$ is a $\WW_{|\umm|,|\unn|}^\#$-valued
analytic function on $D_{g_0} \times D_{1/2}$.
\end{lemma}

\begin{proof}
The analyticity in $g$ follows since $ V_{\umm,\upp,\unn,\uqq}[w^{(I)}(g,\beta,\sigma,z + E_{\rm at})]$ is a polynomial in $g$ and
the coefficients of this polynomial are elements in $\WW_{|\umm|,|\unn|}^\#$ because of \eqref{initial:thmmain:eq2}.
To show the analyticity in $z$ first observe that $V_{\umm,\upp,\unn,\uqq}$ is multilinear expression of integral kernels and that
the kernels $w^{(I)}_{m,n}$ do not depend on $z$ if $m+n \geq 1$. We will use the following algebraic identity
\begin{eqnarray} \label{eq:telescoping2}
\lefteqn{\frac{ A_1(s) \cdots A_n(s) - A_1(s_0) \cdots A_n(s_0) }{s-s_0} } \\
&& - \sum_{i=1}^n A_1(s_0) \cdots A_{i-1}(s_0) A'_i(s_0) A_{i+1}(s_0) \cdots A_n(s_0) \nonumber \\
&&= \sum_{i=1}^n A_1(s) \cdots A_{i-1}(s) \left[ \frac{A_i(s) - A_i(s_0)}{s-s_0} - A'_i(s_0) \right] A_{i+1}(s_0) \cdots A_n(s_0) \nonumber \\
&& + \sum_{i=1}^n \left[ A_1(s) \cdots A_{i-1}(s) - A_1(s_0) \cdots A_{i-1}(s_0) \right] A'_i(s_0) A_{i+1}(s_0) \cdots A_n(s_0) \nonumber .
\end{eqnarray}
Using \eqref{eq:telescoping2} and \eqref{eq:babyestimate1} the analyticity in $z$ follows as a consequence of
the estimates in Lemma \ref{ini:lemelemestimates}
and the following limits for the function $$F_{g,\beta,\sigma}^{(I)}(r) (z) := G_0^{1/2} F[w^{(I)}(g,\beta,\sigma,E_{\rm at} + z)] (H_f + r) G_0^{1/2} . $$
 If $z, z+h \in D_{1/2}$ then for $t=0,1$,
\begin{align*}
& \sup_{r \geq 0} \left\| \frac{1}{h} \partial_r^t \left( F_{g,\beta,\sigma}^{(I)}( z + h)(r) - F_{g,\beta,\sigma}^{(I)}( z )(r) \right) +
\partial_r^t G_0^{1/2}
\frac{\left[\chib^{(I)}(r)\right]^2}{(H_{\rm at} + H_f + r - E_{\rm at} - z )^2 } G_0^{1/2} \right\| \stackrel{ h \to 0 }{\longrightarrow} 0 , \\
&\sup_{r \geq 0} \left\| \partial_r^t  F_{g,\beta,\sigma}^{(I)}( z + h)(r) - \partial_r^t
 F_{g,\beta,\sigma}^{(I)}( z )(r) \right\| \stackrel{ h \to 0 }{\longrightarrow} 0 .
\end{align*}
\end{proof}

\begin{lemma} \label{lem:contwI} The mapping $(\sigma,z) \mapsto V_{\umm,\upp,\unn,\uqq}[w^{(I)}(g,\beta,\sigma,E_{\rm at} + z)]$ is
a $L^2_\omega(\underline{B}_1^{|\umm|+|\unn|} ; C[0,1])$--valued continuous function on $\R_+ \times D_{1/2}$.
\end{lemma}

\begin{proof}
First observe that the kernel $V_{\umm,\upp,\unn,\uqq}$ is a multi-linear expression of integral kernels, thus to show continuity
we can use the following identity,
\begin{eqnarray} \label{eq:telescoping}
\lefteqn{ A_1(s) \cdots A_n(s) - A_1(s_0) \cdots A_n(s_0) } \nonumber \\
&&= \sum_{i=1}^n A_1(s) \cdots A_{i-1}(s) ( A_i(s) - A_i(s_0)) A_{i+1}(s_0) \cdots A_n(s_0) .
\end{eqnarray}
The lemma follows using \eqref{eq:telescoping}, \eqref{eq:babyestimate1}, and the following estimates
\begin{align}
& \| W^{(I)}_{g,\beta}(\sigma_0,z_0) (K^{(m,n)}) - W^{(I)}_{g,\beta}(\sigma,z) (K^{(m,n)}) \|_2 \stackrel{ (\sigma,z) \to (\sigma_0,z_0) }{\longrightarrow} 0 , \label{eq:ccontini1}
\\
&
\sup_{r \geq 0} \left\| F^{(I)}_{g,\beta}(\sigma_0,z_0)(r) - F^{(I)}_{g,\beta}(\sigma,z)(r) \right\| \stackrel{ (\sigma,z) \to (\sigma_0,z_0) }{\longrightarrow} 0 ,
\label{eq:ccontini2}
\end{align}
for the kernels
\begin{align*}
&W^{(I)}_{g,\beta}(\sigma,z) := G_0^{-1/2} \underline{W}_{p,q}^{m,n}[w^{(I)}(g,\beta,\sigma,z + E_{\rm at}) ] G_0^{-1/2} , \\
&F^{(I)}_{g,\beta}(\sigma,z)(r) := G_0^{1/2} F[w^{(I)}(g,\beta,\sigma,z + E_{\rm at})](r + H_f) G_0^{1/2} .
\end{align*}
It remains to show \eqref{eq:ccontini1} and \eqref{eq:ccontini2}. The limit given in \eqref{eq:ccontini2} is
verified by inserting the definitions. Using the notation introduced in \eqref{eq:defofbor} we find
for $m + n + p + q \geq 1$
\begin{eqnarray}
\lefteqn{
\int_{{(\underline{\R}^{3})}^{m+n}} \frac{dK^{(m,n)}}{|K^{(m,n)}|^2} \| B_0(H_f) \underline{W}_{p,q}^{m,n}[w^{}](K^{(m,n)}) B_0(H_f) \|^2 } \nonumber \\
 && \leq
\int_{{(\underline{\R}^3)}^{m+n+p+q}} \frac{d K^{(m,n)}}{|K^{(m,n)}|^2} \frac{d X^{(p,q)}}{|X^{(p,q)}|^2} \nonumber \\
&& \times
\sup_{r \geq 0} \Bigg[ \left\| B_0( r + \Sigma[X^{(p)}]) w_{m+p,n+q}( K^{(m)}, X^{(p)} , \widetilde{K}^{(n)},
 \widetilde{X}^{(q)}) B_0( r + \Sigma[\widetilde{X}^{(q)}]) \right\|^2
 \nonumber \\
 && \times \left( r + \Sigma[{X}^{(p)}]\right)^p \left( r + \Sigma[\widetilde{X}^{(q)}]\right)^q \Bigg] , \nonumber \\
&& =: \left[ \| w \|^\flat_{m,n,p,q} \right]^2 , \label{eq:ccontini1-p1}
 \end{eqnarray}
where we used Lemma \ref{kernelopestimate} and \eqref{eq:trivestimateforkernel}.
Now using dominated convergence it follows from the explicit expression for the kernels
 $w^{(I)}$ that
\begin{equation} \label{eq:ccontini1-p2}
\lim_{ (z,\sigma) \to (z_0,\sigma_0) } \left\| w_{p,q}^{(I)}(g,\beta,\sigma_0,z_0) - w_{p,q}^{(I)}(g,\beta,\sigma,z) \right\|^\flat_{m,n,p,q} = 0 .
\end{equation}
Now \eqref{eq:ccontini1-p1} and \eqref{eq:ccontini1-p2} imply \eqref{eq:ccontini1}.
\end{proof}

\section{Renormalization Transformation}
\label{sec:ren:def}

In this section we define the Renormalization transformation as in \cite{BCFS03} and use results from \cite{HH10}.
Let $0<\xi<1$ and $0 < \rho < 1$.
For $w \in \mathcal{W}_\xi$ we define the analytic function
$$E_\rho[w](z) := \rho^{-1} E[w](z) := - \rho^{-1} w_{0,0}(z,0) = - \rho^{-1} \langle \Omega , H(w(z)) \Omega \rangle$$
 and the set
$$
U[w] := \{ z \in D_{1/2} | | E[w](z) | < \rho / 2 \} .
$$

\begin{lemma} \label{renorm:thm3} Let $0< \rho\leq 1/2$. Then for all $w \in \mathcal{B}(\rho/8, \rho/8, \rho/8 )$, the
function $E_{\rho}[w]: U[w] \to D_{1/2}$ is an analytic bijection.
\end{lemma}
For a proof of the lemma see \cite{BCFS03} or \cite{HH10} (Lemma 21).
In the previous section we introduced smooth functions $\chi_1$ and $\chib_1$. We set
$$
\chi_\rho(\cdot) = \chi_1(\cdot /\rho) \quad , \quad \chib_\rho(\cdot) = \chib_1(\cdot /\rho) \; ,
$$
and use the abbreviation $\chi_\rho = \chi_\rho(H_f)$ and $\chib_\rho = \chib_\rho(H_f)$.
It should be clear from the context whether $\chi_\rho$ or $\chib_\rho$ denotes a function or an operator.

\begin{lemma} \label{renorm:thm1} Let $0 < \rho \leq 1/2$. Then for all $w \in \mathcal{B}(\rho/8,\rho/8,\rho/8)$, and
all $z \in D_{1/2}$ the pair of operators $(H(w(E_\rho[w]^{-1}(z)),H_{0,0}(E_\rho[w]^{-1}(z)))$ is a Feshbach pair for $\chi_\rho$.
\end{lemma}

A proof of Lemma \ref{renorm:thm1} can be found in   \cite{BCFS03}  or  \cite{HH10}  (Lemma 23 and Remark 24).
The definition of the renormalization transformation involves a scaling transformation $S_\rho$ which scales the energy value $\rho$ to the value 1.
It is defined as follows.
For operators $A \in \mathcal{B}(\FF)$ set
$$
S_\rho(A) = \rho^{-1} \Gamma_\rho A \Gamma_\rho^* ,
$$
where $\Gamma_\rho$ is the unitary dilation on $\FF$ which is uniquely determined by
\begin{align*}
& \Gamma_\rho a^\#(k) \Gamma_\rho^* = \rho^{-3/2} a^\#(\rho^{-1} k)  , \quad \Gamma_\rho \Omega = \Omega .
\end{align*}
It is easy to check that $\Gamma_\rho H_f \Gamma_\rho^* = \rho H_f$ and hence $\Gamma_\rho \chi_\rho \Gamma_\rho^* = \chi_1$.
We are now ready to define the renormalization transformation, which in
view of Lemmas \ref{renorm:thm3} and \ref{renorm:thm1} is well defined.

\begin{definition} Let $0 < \rho \leq 1/2$. For $w\in \mathcal{B}(\rho/8,\rho/8,\rho/8)$ we define the renormalization transformation
\begin{equation} \label{eq:defofrenorm}
\left( {R}_\rho H(w) \right)(z) := S_\rho F_{\chi_{\rho}}(H(w(E_\rho[w]^{-1}(z)),H_{0,0}(E_\rho[w]^{-1}(z))) \upharpoonright \HH_{\rm red}
\end{equation}
where $z \in D_{1/2}$.
\end{definition}

\begin{theorem} \label{thm:maingenerala} Let $0<\rho \leq 1/2$ and $0< \xi \leq 1/2$.
For
$w\in \mathcal{B}(\rho/8,\rho/8,\rho/8)$ there exists a unique integral kernel $\mathcal{R}_\rho(w) \in \WW_{\xi}$ such that
\begin{equation} \label{eq:defofrenormkern}
({R}_\rho H(w))(z) = H(\mathcal{R}_\rho(w)(z)) .
\end{equation}
If $w$ is symmetric then also $\mathcal{R}_\rho(w)$ is symmetric. If $w(z)$ is invariant under rotations for all $z \in D_{1/2}$
than also $\mathcal{R}_\rho(w)(z)$ is invariant under rotations for all $z \in D_{1/2}$.
\end{theorem}

A proof of  the existence of the integral kernel as stated in  Theorem \ref{thm:maingenerala} can be found in \cite{BCFS03}
or  \cite{HH10} (Theorem 32). The uniqueness follows from Theorem  \ref{thm:injective}. The statement about  the  rotation invariance
can be seen as follows. If $w(z)$ is rotation invariant for all $z \in D_{1/2}$, then
$H(w(z))$ and $H_{0,0}(w(z))$ and $E_\rho[w](z)$ are rotation
invariant for all $z \in D_{1/2}$, by Lemma \ref{lem:wHinvequiv}.
In that case it follows from the definition of the Feshbach map \eqref{eq:defoffesh} that the right hand side of
\eqref{eq:defofrenorm} is rotation invariant. Now \eqref{eq:defofrenormkern} and  Lemma  \ref{lem:wHinvequiv}
imply that $\mathcal{R}_\rho(w)(z)$ is rotation invariant for all $z \in D_{1/2}$.
The statement about the symmetry
follows from  Lemma  \ref{lem:symmetry}   and the fact  that the Feshbach transformation, the rescaling of the energy, and
reparameterization of the spectral parameter  preserve the  symmetry property.

\begin{theorem} \label{codim:thm1} For any positive numbers $\rho_0 \leq 1/2$ and $\xi_0 \leq 1/2$ there exist numbers $\rho, \xi, \epsilon_0$ satisfying
$\rho \in (0, \rho_0]$, $\xi \in (0, \xi_0]$, and $0 < \epsilon_0 \leq \rho/8$
such that the following property holds,
\begin{equation} \label{codim:thm1:eq}
\mathcal{R}_\rho : \mathcal{B}_0(\epsilon, \delta_1, \delta_2 ) \to
\mathcal{B}_0( \epsilon + \delta_2/2 , \delta_2/2 , \delta_2/2 ) \quad , \quad \forall \ \epsilon, \delta_1, \delta_2 \in [0, \epsilon_0) .
\end{equation}
\end{theorem}

A proof  of Theorem  \ref{codim:thm1} can be found in \cite{HH10} (Theorem 38).
The proof given there  relies on the  fact that there are no terms which are linear in
creation or annihilation operators. Since by rotation invariance and Lemma \ref{lem:wHinvequiv}
there are no terms which
are linear in creation and annihilation operators,  Theorem  \ref{codim:thm1} follows from
the same proof.
Using the contraction property we can iterate the renormalization transformation.  To this end we
introduce the following Hypothesis.

\vspace{0.5cm}

\noindent
(R) \quad Let $\rho, \xi, \epsilon_0$ are positive numbers such that the contraction property \eqref{codim:thm1:eq}
holds and $\rho \leq 1/4$, $\xi \leq 1/4$ and $\epsilon_0 \leq \rho/8$.

\vspace{0.5cm}

Hypothesis (R)
allows us to iterate the renormalization transformation as follows,
\begin{equation} \nonumber
\mathcal{B}_0(\sfrac{1}{2}\epsilon_0 , \sfrac{1}{2} \epsilon_0 , \sfrac{1}{2} \epsilon_0 ) \stackrel{\mathcal{R}_\rho}{\longrightarrow}
\mathcal{B}_0( [ \sfrac{1}{2}+ \sfrac{1}{4} ] \epsilon_0 , \sfrac{1}{4} \epsilon_0 , \sfrac{1}{4} \epsilon_0 ) \stackrel{\mathcal{R}_\rho}{\longrightarrow} \cdots
\mathcal{B}_0( \Sigma_{l=1}^n \sfrac{1}{2^l} \epsilon_0 , \sfrac{1}{2^n} \epsilon_0 , \sfrac{1}{2^n} \epsilon_0 )  \stackrel{\mathcal{R}_\rho}{\longrightarrow} \cdots   .
\end{equation}

\begin{theorem} \label{thm:bcfsmain} Assume Hypothesis (R).
There exist functions
\begin{align*}
&e_{(0)}[ \cdot ] : \mathcal{B}_0(\epsilon_0/2,\epsilon_0/2,\epsilon_0/2) \to D_{1/2} \\
&\psi_{(0)}[ \cdot ] : \mathcal{B}_0(\epsilon_0/2,\epsilon_0/2,\epsilon_0/2) \to \FF
\end{align*}
such that the following holds.
\begin{itemize}
\item[(a)] For all $w \in \mathcal{B}_0(\epsilon_0/2,\epsilon_0/2,\epsilon_0/2)$,
$$
{\rm dim} \ker \{ H(w(e_{(0)}[w]) \} \geq 1 ,$$
and $\psi_{(0)}[w]$ is a nonzero element in the kernel of $H(w(e_{(0)}[w])$.
\item[(b)] If $w$ is symmetric and $-1/2 < z < e_{(0)}[w]$, then $H(w(z))$ is bounded invertible.
\item[(c)] The function $\psi_{(0)}[ \cdot ]$ is uniformly bounded with bound
$$\sup_{w \in \mathcal{B}_0(\epsilon_0/2,\epsilon_0/2,\epsilon_0/2)} \| \psi_{(0)}[w] \| \leq 4 e^4 . $$
\item[(d)] Let $S$ be an open subset of $\C$ respectively a topological space. Suppose
\begin{eqnarray*}
w(\cdot, \cdot) : &&S \times D_{1/2} \to \WW_\xi^\# \\
&&(s,z) \mapsto w(s,z)
\end{eqnarray*}
is an analytic respectively a c-continuous function such that $w(s)(\cdot) := w(s, \cdot)$ is in $\mathcal{B}_0(\epsilon_0/2,\epsilon_0/2,\epsilon_0/2)$.
Then $s \mapsto e_{(0)}[w(s)]$ and $\psi_{(0)}[w(s)]$ are analytic respectively continuous functions.
\end{itemize}
\end{theorem}

A proof of Theorem
\ref{thm:bcfsmain} is given in  \cite{HH10} (Theorem 42 and Theorem 43). We want to note that the
proof which can be found there of part (a) of Theorem \ref{thm:bcfsmain}  is from   \cite{BCFS03}.

\section{Main Theorem}

\label{sec:prov}

In this section, we prove Theorem \ref{thm:main1}, the main result of this paper.
Its proof is based on Theorems \ref{thm:inimain1} and \ref{thm:bcfsmain}.

\vspace{0.5cm}

\noindent
{\it Proof of Theorem \ref{thm:main1}.} Choose $\rho, \xi, \epsilon_0$ such that Hypothesis (R) holds. Choose $g_0$ such that the conclusions
of Theorem \ref{thm:inimain1} hold for $\delta_1=\delta_2=\delta_3 = \epsilon_0/2$. Let   $ g \in D_{g_0}$. It  follows from
Theorem \ref{thm:bcfsmain} (a)
 that
$\psi_{(0)}[w^{(0)}(g,\beta,\sigma)]$ is a nonzero element in the kernel of $H_{g,\beta,\sigma}^{(0)}(e_{(0)}[w^{(0)}(g,\beta,\sigma)])$.
From Theorem \ref{thm:inimain1} we know that there exists a finite $C_Q$ such that
\begin{equation} \label{eq:qdef11}
\sup_{ (g,\beta,\sigma,z) \in B_0 \times \R \times \R_+ \times D_{1/2} } | Q_{\chi^{(I)}}(g,\beta,\sigma,z) \| \leq C_Q .
\end{equation}
 From the Feshbach property, Theorem  \ref{thm:fesh}, it follows that
\begin{equation} \label{eq:psiis}
\psi_{\beta,\sigma}(g) := Q_{\chi^{(I)}}(g,\beta,\sigma, e_{(0)}[w^{(0)}(g,\beta,\sigma)]) \psi_{(0)}[w^{(0)}(g,\beta,\sigma)]  ,
\end{equation}
is nonzero and an eigenvector of $H_{g,\beta,\sigma}$ with eigenvalue $E_{\beta,\sigma}(g) := E_{\rm at} + e_{(0)}[w^{(0)}(g,\beta,\sigma)] $.
By Theorem \ref{thm:inimain1}, we know that $(g,z) \mapsto w^{(0)}(g,\beta,\sigma,z)$ is analytic and hence by Theorem \ref{thm:bcfsmain} (d) it follows that
 $g \mapsto \psi_{(0)}[w^{(0)}(g,\beta,\sigma)] $ and $g \mapsto E_{\beta,\sigma}(g)$
are analytic. This implies that $g \mapsto \psi_{\beta,\sigma}(g)$ is analytic because of the analyticity
of $(g,z) \mapsto Q_{\chi^{(I)}}(g,\beta,\sigma,z)$ and  \eqref{eq:psiis}.
By Theorem \ref{thm:inimain1}, we know that $(\sigma,z) \mapsto w^{(0)}(g,\beta,\sigma,z)$ is c-continuous.
By Theorem \ref{thm:bcfsmain} (d) it now follows that $\sigma \mapsto \psi_{(0)}[w^{(0)}(g,\beta,\sigma)] $ and $\sigma \mapsto E_{\beta,\sigma}(g)$
are continuous. This implies that $\sigma \mapsto \psi_{\beta,\sigma}(g)$ is continuous because of the continuity of
of $(\sigma,z) \mapsto Q_{\chi^{(I)}}(g,\beta,\sigma,z)$ and  \eqref{eq:psiis}.
As a consequence of the  definition it follows that we have the bound
\begin{equation} \label{eq:boundonE11}
\sup_{(g,\beta,\sigma) \in D_{g_0} \times \R \times \R_+} | E_{\beta,\sigma}(g) | \leq E_{\rm at} + 1/2 .
\end{equation}
By \eqref{eq:psiis}, Theorem \ref{thm:bcfsmain} (c), and the bound in \eqref{eq:qdef11} we have
\begin{equation} \label{eq:boundonpsi11}
\sup_{ (g,\beta,\sigma) \in D_{g_0} \times \R \times \R_+ } | \psi_{\beta,\sigma}(g) | \leq C_Q 4 e^4 .
\end{equation}
By possibly restricting to a smaller ball than $D_{g_0}$ we can ensure that the projection operator
\begin{equation} \label{eq:projectionlambda0}
P_{\sigma,\beta}(g)
:= \frac{ \left| \psi_{\beta,\sigma}(g) \right\rangle \left\langle \psi_{\beta,\sigma}(\overline{g}) \right| }{ \left\langle \psi_{\beta,\sigma}(\overline{g}) , \psi_{\beta,\sigma}({g}) \right\rangle } ,
\end{equation}
is well defined for all $( \beta,\sigma) \in \R \times \R_+$ and $g \in D_{g_0}$,
which is shown as follows. First observe that the denominator of \eqref{eq:projectionlambda0} is an analytic
function of $g$. By fixing the normalization we can assume that
$\langle \psi_{\beta,\sigma}(0) , \psi_{\beta,\sigma}(0) \rangle = 1$. If we estimate
the remainder of the Taylor expansion of the denominator of \eqref{eq:projectionlambda0} using analyticity and
the uniform bound \eqref{eq:boundonpsi11} it follows, by possibly choosing $g_0$ smaller but still positive, that there exists a positive constant $c_0$ such that
$| \langle \psi_{\beta,\sigma}(\overline{g}) , \psi_{\beta,\sigma}({g}) \rangle | \geq c_0 $
for all $|g| \leq {g}_0$.
 Now
using the corresponding property of $\psi_{\beta,\sigma}(g)$, it follows from \eqref{eq:projectionlambda0} that
$P_{\beta,\sigma}(g)$ is analytic on $D_{g_0}$, continuous in $\sigma$ and that
\begin{equation} \label{eq:boundonP}
\sup_{(g,\beta,\sigma) \in D_{g_0} \times \R \times \R_+} \| P_{\sigma,\beta}(g) \| < \infty .
\end{equation}
If $g \in D_{g_0} \cap \R$, then by definition \eqref{eq:projectionlambda0} we see that
$P_{\beta,\sigma}(g)^* = P_{\beta,\sigma}(\overline{g})$.

The kernel
$w^{(0)}(g,\beta,\sigma)$ is symmetric for $g \in D_{g_0} \cap \R$, see Theorem \ref{thm:inimain1}.
It now follows from Theorem \ref{thm:bcfsmain} (b) that $H^{(0)}_{g,\beta,\sigma}(z)$ is bounded invertible if $z \in (-\frac{1}{2}, e_{(0,\infty)}[w^{(0)}(g,\beta,\sigma)] )$. Applying the Feshbach property, Theorem  \ref{thm:fesh},
it follows that $H_{g,\beta,\sigma} - \zeta$ is bounded invertible for $\zeta \in ( E_{\rm at} -\frac{1}{2}, E_{\rm at} + e_{(0,\infty)}[w^{(0)}(g,\beta,\sigma)] )$.
For $\zeta \leq E_{\rm at} - 1/2$ the bounded invertibility of $H_{g,\beta,\sigma} - \zeta$ for $g$ sufficiently small follows from
the estimate $$\| ( H_0 - \zeta )^{-1} W_{g,\beta,\sigma} \| \leq 4 \| ( H_0 - E_{\rm at} + 2 )^{-1} W_{g,\beta,\sigma} \| \leq C|g|,$$
where in the first inequality we used that $E_{\rm at}$ is the infimum of the spectrum of $H_0$ and in the second inequality we used
the estimate of the second factor in \eqref{eq:FestimateAA}, which is given in the proof of Theorem \ref{ini:thm1}.
Thus $E_{\beta,\sigma}(g) = \inf \sigma(H_{g,\beta,\sigma})$ for real $g \in D_{g_0} \cap \R$.
\qed

\vspace{0.5cm}

We want to note that the proof provides an  explicit  bound on the ground state energy, Eq.  \eqref{eq:boundonE11}.
Next we show that Theorem \ref{thm:main1} implies Corollary \ref{cor:main1}.

\vspace{0.5cm}

\noindent
{\it Proof of Corollary \ref{cor:main1}.}
We use Cauchy's formula.
 For any positive $r$ which is less than $g_0$, we have
\begin{equation} \label{eq:cauchy223}
E_{\beta,\sigma}^{(n)} = \frac{1 }{2 \pi i} \int_{|z|= r} \frac{E_{\beta,\sigma}(z)}{z^{n+1}} d z , \
\psi_{\beta,\sigma}^{(n)} = \frac{1 }{2 \pi i} \int_{|z|= r} \frac{\psi_{\beta,\sigma}(z)}{z^{n+1}} d z , \
P_{\beta,\sigma}^{(n)} = \frac{1 }{2 \pi i} \int_{|z|= r} \frac{P_{\beta,\sigma}(z)}{z^{n+1}} d z .
\end{equation}
The first equation of  \eqref{eq:cauchy223}
 implies  that $|E_{\beta,\sigma}^{(n)}| \leq r^{-n} \sup_{(g,\beta,\sigma) \in D_{g_0} \times \R \times \R_+} | E_{\beta,\sigma}(g) |$
and that $\sigma \mapsto E_{\beta,\sigma}^{(n)}$ is continuous on $\R_+$ by dominated convergence.
Similarly we conclude by \eqref{eq:cauchy223}
that there exists a finite constant $C$ such that $\| \psi_{\beta,\sigma}^{(n)}\| \leq C r^{-n}$,
respectively $\|P^{(n)}_{\beta,\sigma} \| \leq C r^{-n}$,
and that $ \psi_{\beta,\sigma}^{(n)}$, respectively $P_{\beta,\sigma}^{(n)}$, are continuous functions of $ \sigma \in \R_+$.
Finally observe that $(-1)^N H_{g,\beta,\sigma}(-1)^N = H_{-g,\beta,\sigma}$ where $N$ is the closed linear operator on $\FF$ with $N \upharpoonright \FF^{(n)}(\hh) = n$.
This implies that the ground state energy $E_{\beta,\sigma}(g)$ cannot depend on odd powers of $g$.
\qed

\section*{Acknowledgements}

 D. H. wants to thank J. Fr\"ohlich, G.M. Graf, M. Griesemer, and A. Pizzo for interesting conversations.
Moreover, D.H. wants to thank ETH Zurich   for hospitality and financial support.

\section*{Appendix A: Elementary Estimates and the Pull-through Formula}
\label{sec:appA}

To give a precise meaning to expressions which occur in \eqref{eq:defhmn11} and \eqref{eq:defhlinemn}, we introduce the following.
For $\psi \in \FF$ having finitely many particles we have
\beqn \label{eq:defofa}
\left[ a(K_1) \cdots a(K_m) \psi \right]_n(K_{m+1},...,K_{m+n}) = \sqrt{\frac{(m+n)!}{n!}} \psi_{m+n}(K_{1},...,K_{m+n}) ,
\eeqn
for all $K_1,...,K_{m+n} \in \underline{\R}^3 := \R^3 \times \Z_2$, and
using Fubini's theorem it is elementary to see that the vector valued map
 $(K_1,...,K_m) \mapsto a(K_1) \cdots a(K_m) \psi$ is
an element of $L^2((\underline{\R}^{3})^m; \FF)$. The following lemma states the well
known pull-through formula.
For a proof see for example \cite{BFS98,HH10}.
\begin{lemma} \label{lem:pullthrough}
Let $f : \R_+ \to \C$ be a bounded measurable function. Then for all $K \in \R^3 \times \Z_2$
$$
f(H_f) a^*(K) = a^*(K) f(H_f + \omega(K) ) , \quad a(K) f(H_f) = f(H_f + \omega(K) ) a(K) .
$$
\end{lemma}
Let $w_{m,n}$ be function on $\R_+ \times ({\underline{\R}^{3})}^{n+m}$ with values
in the linear operators of $\HH_{\rm at}$ or the complex numbers.
To such a function we associate the quadratic form
\begin{eqnarray*}
q_{w_{m,n}}(\varphi,\psi) := \int_{{(\underline{\R}^3)}^{m+n}} \frac{ dK^{(m,n)}}{|K^{(m,n)}|^{1/2}}
 \left\langle a(K^{(m)}) \varphi ,w_{m,n}(H_f, K^{(m,n)}) a(\widetilde{K}^{(n)}) \psi \right \rangle ,
\end{eqnarray*}
defined  for all $\varphi$ and $\psi$ in $\HH$ respectively $\FF$, for which the right hand side is defined as a complex number.
To associate an operator to the quadratic form we will use the following lemma.
\begin{lemma} \label{kernelopestimate} Let $\underline{X} = \R^3 \times \Z_2$. Then
\begin{eqnarray} \label{eq:defofH}
| q_{w_{m,n}}(\varphi,\psi) | \leq \| w_{m,n} \|_\sharp \| \varphi \| \| \psi \| ,
\end{eqnarray}
where
\begin{eqnarray*}
 \| w_{m,n} \|_\sharp^2 :=
\int_{\underline{X}^{m+n}} \frac{d K^{(m,n)}}{|K^{(m,n)}|^2} \sup_{r \geq 0} \left[ \|w_{m,n}(r,K^{(m,n)}) \|^2 \prod_{l=1}^m \left\{ r + \Sigma[K^{(l)}] \right\}
 \prod_{\widetilde{l}=1}^n \left\{ r + \Sigma[\widetilde{K}^{(\widetilde{l})}] \right\} \right] .
\end{eqnarray*}
\end{lemma}
\begin{proof} We set $P[K^{(n)}] := \prod_{l=1}^n ( H_f + \Sigma[K^l])^{1/2}$ and insert 1's to obtain the trivial identity
\begin{align*}
 | q_{w_{m,n}}(\varphi,\psi) | &=
\Bigg| \int_{\underline{X}^{m+n}} \frac{d K^{(m,n)}}{|K^{(m,n)}|}
 \Big\langle
 P[K^{(m)}] P[K^{(m)}]^{-1} |K^{(m)}|^{1/2} a(K^{(m)}) \varphi , w_{m,n}(H_f ,K^{(m,n)})
 \\
&
\times
P[\widetilde{K}^{(n)}] P[\widetilde{K}^{(n)}]^{-1} | \widetilde{K}^{(n)}|^{1/2} a(\widetilde{K}^{(n)}) \psi \Big\rangle \Bigg| .
\end{align*}
The lemma now follows using the Cauchy-Schwarz
inequality and the following well known identity for $n \geq 1$ and $\phi \in \FF$,
\begin{eqnarray}
\int_{\underline{X}^n}
d K^{(n)} | K^{(n)} | \left\| \prod_{l=1}^n \left[ H_f + \Sigma[K^{(l)}] \right]^{-1/2} a(K^{(n)}) \phi \right\|^2 = \| P_\Omega^\perp \phi \|^2 \label{eq:trivialA}  ,
\end{eqnarray}
where
$P_\Omega^\perp := | \Omega \rangle \langle \Omega |$.
A proof of \eqref{eq:trivialA}  can for example be found in \cite{HH10} Appendix A.
\end{proof}

Provided the form $q_{w_{m,n}}$ is densely defined and $ \| w_{m,n} \|_\sharp$ is a finite real number,
then the form $q_{w_{m,n}}$ determines uniquely a bounded
linear operator  $\underline{H}_{m,n}(w_{m,n})$ such that
 $$
q_{w_{m,n}}(\varphi,\psi ) = \langle \varphi,\underline{H}_{m,n}(w_{m,n}) \psi \rangle ,
$$
for all $\varphi, \psi$ in the form domain of $q_{w_{m,n}}$. Moreover,
$\| \underline{H}_{m,n}(w_{m,n}) \|_{} \leq \| w_{m,n} \|_\sharp$.
Using the pull-through formula and Lemma \ref{kernelopestimate} it is easy to see
that for $w^{(I)}$, defined in \eqref{defofwI}, with $m+n=1,2$, the form
$$
q^{(I)}_{m,n}(\varphi, \psi) := q_{w_{m,n}^{(I)}}( \varphi , (H_f + 1 )^{-\frac{1}{2}(m+n)} (-\Delta + 1 )^{-\frac{1}{2} \delta_{1,m+n}} \psi )
$$
is densely defined and bounded.
Thus we can associate a bounded linear operator $L_{m,n}^{(I)}$ such that
$q_{m,n}^{(I)}(\varphi, \psi) = \langle \varphi , L_{m,n}^{(I)} \psi \rangle$. This allows us to define
$$
\underline{H}_{m,n}(w_{m,n}^{(I)}) := L_{m,n}^{(I)} (H_f + 1 )^{\frac{1}{2}(m+n)} (-\Delta + 1 )^{\frac{1}{2}\delta_{1,m+n}}
$$
as an operator in $\HH$.

\section*{Appendix B: Generalized Wick Theorem}
\label{sec:appB}

For $m,n \in \N_0$ let  $\underline{\mathcal{M}}_{m,n}$ denote the space of measurable functions on $\R_+ \times (\underline{\mathbb{R}}^3)^{m+n}$ with values
in the linear operators of $\HH_{\rm at}$. Let
$$
\underline{\mathcal{M}} = \bigoplus_{m+n=1,2} \underline{\mathcal{M}}_{m,n} .
$$
For  $w \in \underline{\mathcal{M}}$ we define
$$
\underline{W}[w] := \sum_{m+n=1,2} \underline{H}_{m,n}(w) .
$$
The following Theorem is from \cite{BFS98}. It is a generalization of Wick's Theorem.
\begin{theorem} \label{thm:wicktheorem} Let $w \in \underline{\mathcal{M}}$
 and let $F_0,F_1,...,F_L \in   \underline{\mathcal{M}}_{0,0}     $. Then as a formal identity
$$
F_0(H_f) \underline{W}[w] F_1(H_f) \underline{W}[w] \cdots \underline{W}[w] F_{L-1}(H_f) \underline{W}[w] F_L(H_f) = \underline{H}( \widetilde{w}^{({\rm sym})} ) ,
$$
where
\begin{eqnarray}
\lefteqn{ \widetilde{w}_{M,N}(r;K^{(M,N)}) } \nonumber \\
& = &
\sum_{\substack{ m_1 + \cdots m_L = M \\ n_1+...n_L=N }} \sum_{\substack{ p_1, q_1,...,p_L,q_L: \\ m_l+p_l+n_l+q_l \geq 1 }} \prod_{l=1}^L
\left\{ \binom{ m_l + p_l}{ p_l} \binom{ n_l + q_l}{ q_l } \right\}
\nonumber \\
& & \times F_0(r + \tilde{r}_0)
\langle \Omega , \prod_{l=1}^{L-1} \left\{
\underline{W}_{p_l,q_l}^{m_l,n_l}[w]( r + r_l ; K_l^{(m_l,n_l)}) F_{l}(H_f + r + \widetilde{r}_{l}) \right\} \nonumber \\
& &
\underline{W}_{p_L,q_L}^{m_L,n_L}[w]( r + r_L ; K_L^{(m_L,n_L)}) \Omega \rangle F_L(r + \widetilde{r}_L) , \label{eq:complicated}
\end{eqnarray}
with
\begin{align}
& K^{(M,N)} := (K_1^{(m_1,n_1)}, ... , K_L^{(m_L,n_L)}) , \quad K_l^{(m_l,n_l)} := (k_l^{(m_l)},\widetilde{k}_l^{(n_l)}) , \label{eq:KMNdef}
\\
& r_l := \Sigma[\widetilde{K}_1^{(n_1)}] + \cdots + \Sigma[\widetilde{K}_{l-1}^{(n_{l-1})}] + \Sigma[{K}_{l+1}^{(m_{l+1})}] + \cdots + \Sigma[{K}_L^{(m_L)}] , \label{eq:rldef}
\\
& \widetilde{r}_l := \Sigma[\widetilde{K}_1^{(n_1)}] + \cdots + \Sigma[\widetilde{K}_{l}^{(n_{l})}] + \Sigma[{K}_{l+1}^{(m_{l+1})}] + \cdots + \Sigma[{K}_L^{(m_L)}] . \label{eq:rltildedef}
\end{align}
\end{theorem}
A proof can be found in \cite{BFS98}. We note that the proof is essentially the same as the proof of Theorem 3.6 in \cite{BCFS03} or Theorem 27 in
\cite{HH10}.

\section*{Appendix C: Smooth Feshbach Property}
\label{sec:smo}

In this appendix we follow \cite{BCFS03,GH08}.
We introduce the Feshbach map and state basic isospectrality
properties.
Let $\chi$ and $\overline{\chi}$ be commuting, nonzero bounded operators, acting on a separable Hilbert space $\HH$
and satisfying $\chi^2 + \overline{\chi}^2=1$. A {\it Feshbach pair} $(H,T)$ for $\chi$ is a pair of
closed operators with the same domain,
$$
H,T : D(H) = D(T) \subset \HH \to \HH
$$
such that $H,T, W := H-T$, and the operators
\begin{align*}
&W_\chi := \chi W \chi , & &W_{\overline{\chi}} := \overline{\chi} W \chib \\
&H_\chi :=T + W_\chi , & &H_{\overline{\chi}} := T + W_{\chib} ,
\end{align*}
defined on $D(T)$ satisfy the following assumptions:
\begin{itemize}
\item[(a)] $\chi T \subset T \chi$ and $\chib T \subset T \chib$,
\item[(b)] $T, H_{\chib} : D(T) \cap \ran \chib \to \ran \chib$ are bijections with bounded inverse,
\item[(c)] $\chib H_{\chib}^{-1} \chib W \chi : D(T) \subset \HH \to \HH$ is a bounded operator.
\end{itemize}
\begin{remark} \label{rem:abuse} {\em
By abuse of notation we write $ H_{\chib}^{-1} \chib$ for $ \left( H_{\chib} \upharpoonright \ran \chib \right)^{-1} \chib$ and
likewise $ T^{-1} \chib$ for $ \left( T \upharpoonright \ran \chib \right)^{-1} \chib$. }
\end{remark}
We call an operator $A:D(A) \subset \HH \to \HH$ {\it bounded invertible} in a subspace $V \subset \HH$
($V$ not necessarily closed), if $A: D(A) \cap V \to V$ is a bijection with bounded inverse.
Given a Feshbach pair $(H,T)$ for $\chi$, the operator
\begin{align} \label{eq:defoffesh}
&F_\chi(H,T) := H_\chi - \chi W \chib H_{\chib}^{-1} \chib W \chi
\end{align}
on $D(T)$ is called the { \it Feshbach map of} $H$.
The auxiliary operator 
\begin{align} \label{eq:defofQ}
  Q_\chi := Q_\chi(H,T) := \chi - \chib H_{\chib}^{-1} \chib W \chi
\end{align}
is by conditions (a), (c), bounded, and $Q_\chi$ leaves $D(T)$ invariant. The Feshbach map is
isospectral in the sense of the following theorem.
\begin{theorem} \label{thm:fesh}
Let $(H,T)$ be a Feshbach pair for $\chi$ on a Hilbert space $\HH$. Then the following holds.
$\chi \ker H \subset \ker F_\chi(H,T)$ and $Q_\chi \ker F_\chi(H,T) \subset \ker H$. The mappings
\begin{align*}
\chi : \ker H \to \ker F_\chi(H,T) , \quad Q_\chi : \ker F_\chi(H,T) \to \ker H ,
\end{align*}
are linear isomoporhisms and inverse to each other. $H$ is bounded invertible on $\HH$ if and only if $F_\chi(H,T)$ is
bounded invertible on $\ran \chi$.
\end{theorem}

The proof of Theorem \ref{thm:fesh} can be found in \cite{BCFS03,GH08}. The next lemma
gives sufficient conditions for  two operators to be a Feshbach pair. It follows
from a Neumann expansion, \cite{GH08}.

\begin{lemma} \label{fesh:thm2}
Conditions {\rm (a), (b)}, and {\rm (c)} on Feshbach pairs are satisfied if:
\begin{itemize}
\item[(a')] $\chi T \subset T \chi$ and $\chib T \subset T \chib$,
\item[(b')] $T$ is bounded invertible in $\ran \chib$,
\item[(c')] $\| T^{-1} \chib W \chib \| < 1$, $\| \chib W T^{-1} \chib \| < 1$, and $T^{-1} \chib W \chi$ is a bounded operator.
\end{itemize}
\end{lemma}


\begin{thebibliography}{30}







\bibitem{BCFS03} V. Bach, T. Chen, J. Fr\"ohlich, I.M. Sigal,
{\em Smooth Feshbach map and operator-theoretic renormalization group methods},
 J. Funct. Anal. 203 (2003), 44--92.

\bibitem{BFP06} V. Bach, J. Fr\"ohlich, A. Pizzo, {\it Infrared-finite Algorithms in QED: the groundstate of an Atom interacting with the quantized radiation field},
Comm. Math. Phys. 264 (2006),  145--165.

\bibitem{BFP07} V. Bach, J. Fr\"ohlich, A. Pizzo, {\it
An infrared-finite algorithm for Rayleigh scattering amplitudes, and Bohr's frequency condition},
Comm. Math. Phys. 274 (2007), 457--486.

\bibitem{BFP09} V. Bach, J. Fr\"ohlich, A. Pizzo, {\it Infrared-finite Algorithms in QED. II. The expansion of the groundstate of an Atom interacting with the quantized radiation field},
Adv. Math. 220 (2009),  1023--1074.



\bibitem{BFS98} V. Bach, J. Fr\"ohlich, I.M. Sigal,
{\em Renormalization group analysis of spectral problems in quantum field theory},
 Adv. Math. 137 (1998), 205--298.


\bibitem{BFS99} V. Bach, J. Fr\"ohlich, I.M. Sigal, {\it Spectral analysis for systems of atoms and molecules coupled to the quantized radiation field},
 Comm. Math. Phys. 207 (1999),  249--290.

\bibitem{BCVV09} J-M. Barbaroux, T. Chen, S. Vugalter, V. Vougalter
{\it Quantitative estimates on the Hydrogen ground state energy in non-relativistic QED }. {\tt mp\_arc 09-48 }



\bibitem{bet47} H.A.~Bethe, The Electromagnetic Shift of Eneregy Levels, Phys. Rev.  72
(1947), 339--341.





\bibitem{HS02} C. Hainzl, R. Seiringer, {\it
Mass renormalization and energy level shift in non-relativistic QED},  Adv. Theor. Math. Phys.   6 (2002),   847--871.

\bibitem{HHS05} C. Hainzl, M. Hirokawa, H. Spohn, {\em
Binding energy for hydrogen-like atoms in the Nelson model without cutoffs},
J. Funct. Anal. 220 (2005), 424--459.



\bibitem{S04} H. Spohn, {\it
Dynamics of charged particles and their radiation field}, Cambridge University Press, Cambridge, 2004.

\bibitem{GLL01}
M. Griesemer, E. Lieb, M. Loss, {\it
Ground states in non-relativistic quantum electrodynamics}.
Invent. Math. 145 (2001), 557--595.

\bibitem{GH08} M. Griesemer, D. Hasler, {\em On the smooth Feshbach-Schur Map},
J. Funct. Anal. 254 (2008), 2329--2335.

\bibitem{GH09} M. Griesemer, D. Hasler,
{\em Analytic perturbation theory and renormalization analysis of matter coupled to quantized radiation},
 Ann. Henri Poincar\'{e}  10  (2009),   577--621.


\bibitem{HH08} D. Hasler, I. Herbst,
{\em On the self-adjointness and domain of Pauli-Fierz type Hamiltonians}.
Rev. Math. Phys. 20 (2008), 787--800.



\bibitem{HH10} D. Hasler, I. Herbst, {\it Ground states in  the spin boson model}, submitted.  {\tt arXiv:1003.5923 }

\bibitem{HHH08} D. Hasler, I. Herbst, M. Huber, {\it
On the lifetime of quasi-stationary states in non-relativistic QED},
Ann. Henri Poincar\'{e} 9 (2008), 1005--1028.

\bibitem{H02} F. Hiroshima, {\em Self-adjointness of the Pauli--Fierz Hamiltonian for arbitrary values of coupling constants}, Ann. Henri Poincar\'{e}
3(1) (2002), 171--201.



\bibitem{lamret47}
W.E. Lamb, R.C. Retherford, Fine Structure of the Hydrogen Atom by a Microwave Method,
Phys. Rev. 72  (1947), 241--243.

\bibitem{reesim2} M. Reed and B. Simon, {\it Methods of modern mathematical physics. II. Fourier-Analysis, Self-Adjointness},
Academic Press, New York-London, 1978.


\bibitem{S07} I.M. Sigal,
{\em Ground state and resonances in the standard model of the non-relativistic QED},  J. Stat. Phys. 134 (2009),  5-6, 899--939.












\end{thebibliography}
\end{document}